\documentclass[11pt]{article}
 \usepackage{graphicx}
\usepackage{amsmath}
\usepackage{amsfonts}
\usepackage{hyperref}
\usepackage{amsthm}
 \usepackage{amscd}
\usepackage{mathrsfs}
\usepackage[top=1in,left=1in,right=1in,bottom=1in]{geometry}
\usepackage{setspace}
\usepackage{xcolor}
\newtheorem{theorem}{Theorem}[section]

\newtheorem{example}{Example}[section]
\newcommand{\ds}{\displaystyle}
\hypersetup{
	colorlinks = true,
	linkcolor  = red,
	citecolor  = green,
	filecolor  = cyan,
	urlcolor   = magenta,}

\numberwithin{equation}{section}

\begin{document}
\title{A Generalized Method for the Darboux Transformation}

\author{Tuncay Aktosun\\
Department of Mathematics\\
University of Texas at Arlington\\
Arlington, TX 76019-0408, USA\\
\\
Mehmet Unlu\\
Department of Mathematics\\
Recep Tayyip Erdogan University\\
53100 Rize, Turkey}

\date{}
\maketitle

\noindent {\bf Abstract}:
A method is presented to obtain the change
in the potential and in the relevant wavefunction of a linear system of ordinary differential equations containing a spectral parameter,
when that linear system is perturbed and a finite number of discrete eigenvalues
are added to or removed from the spectrum. Some explicit formulas are derived for those changes
by introducing certain fundamental linear integral equations for the corresponding unperturbed and perturbed linear systems.
This generalized method is applicable in a unified manner on a wide class of linear systems. This is in contrast to
the standard method for a Darboux transformation, which is specific to the particular linear system on which it applies.
A comparison is provided in some special cases between this generalized method and the standard method
for the Darboux transformation. In particular, when a bound state is added to the
discrete spectrum, some Darboux transformation formulas are presented
for the full-line Schr\"odinger equation, where those formulas resemble the Darboux transformation formulas 
for the half-line Schr\"odinger equation. The theory presented is illustrated
with some explicit examples.

\vfill


\newpage
\section{Introduction}
\label{section1}

In this paper we consider the Darboux transformation for the spectral problem  $\mathcal L\Psi = \lambda\Psi,$
where $\mathcal L$ is a linear ordinary differential operator acting on some function space, $x$ is the independent spacial
variable, and $\lambda$ is the spectral parameter.
The spectrum of $\mathcal L$ consists of all $\lambda$-values for which there exists a
nonzero solution $\Psi(\lambda,x),$ which is called a wavefunction.
The operator $\mathcal L$ usually contains a function $u(x)$ as a coefficient, which is called the potential.
There are infinitely many choices of wavefunctions for a given spectral problem.
Certain specific wavefunctions are more convenient to use in the description of the Darboux transformation,
and the choice depends on the particular scattering or
spectral data used as input to the Darboux transformation.

Let us
perturb the operator $\mathcal L$ to $\tilde{\mathcal L}$ in a such a way that the continuous spectra of
$\mathcal L$ and $\tilde{\mathcal L}$ coincide, while their discrete spectra differ by a set consisting of
a finite number of eigenvalues. 
Such a perturbation is known as a Darboux transformation after the
French mathematician Gaston Darboux \cite{D1882}. Thus,
as a result of the perturbation, the unperturbed spectral problem
$\mathcal L\Psi = \lambda \Psi$ is changed to the perturbed problem $\tilde{\mathcal L} \tilde\Psi= \lambda \tilde\Psi,$
and we have the transformations $\mathcal L \mapsto
\tilde{\mathcal L}$, $u(x) \mapsto \tilde u(x),$ and
$\Psi(\lambda,x) \mapsto \tilde\Psi(\lambda,x)$ for the operator,
the potential, and the wavefunction, respectively.
Note that we use a tilde to denote the corresponding perturbed quantity.

We can view the Darboux transformation as having two parts: the first is at the potential level and the second is at the wavefunction
level. At the potential level, the Darboux transformation consists of determining $\tilde u(x)$
in terms of $u(x)$ and the quantities evaluated at the discrete $\lambda$-eigenvalues appearing
in the perturbation. At the wavefunction level, it involves the
determination of $\tilde\Psi(\lambda,x)$ in terms of $\Psi(\lambda,x)$ and the quantities evaluated
at the discrete $\lambda$-eigenvalues related to the perturbation.

A Darboux transformation naturally occurs as a special case in inverse problems. In an inverse problem, a
potential is recovered from the corresponding scattering or spectral data set. In the case
of a Darboux transformation, the solution to the inverse problem involves the recovery of the change
in the potential from the corresponding change
occurring only in the discrete spectrum.
A Darboux transformation often yields a closed-form solution to the 
corresponding inverse problem because the data set used contains a finite-rank
perturbation of the spectral data. Furthermore, a Darboux transformation 
may yield a closed-form explicit solution to a differential equation in order
to produce
other explicit solutions to the same differential equation or to other related differential equations.
For various applications and a historical account of
Darboux transformations, we refer the reader to many available references such as
\cite{CS1989,C1955,DT1979,GHZ2005,MS1991,RS2002}
and the references therein.

The Darboux transformation formulas available in the literature, which we refer to as the standard
Darboux transformation, are designed to be applicable to certain specific differential equations
and to particular spacial domains for those differential equations. For example,
the standard Darboux transformation for the full-line Schr\"odinger equation
and the standard Darboux transformation for the half-line Schr\"odinger equation
are obtained differently, and the corresponding Darboux transformation formulas
are also different. We elaborate on this difference at the ends of
Sections~\ref{section2}, \ref{section3}, and \ref{section4}.

Our main goal in this paper is to present a method to obtain the Darboux
transformation for a wide class of linear differential operators by using the same process. Because we use the same procedure on different
differential equations and on different spacial domains, we refer to our method as the generalized method for
the Darboux transformation. We accomplish our goal by introducing a corresponding
fundamental linear integral equation for each of the unperturbed and perturbed linear operators,
and this procedure enables us
to describe the Darboux transformation from a unified viewpoint, regardless of the particular
differential equation and of its spacial domain under consideration.
We use the difference between the kernels of the
perturbed and unperturbed integral operators as input to
the Darboux transformation. We then show that the perturbed fundamental linear integral equation can be transformed
into a modified linear integral equation with a separable kernel. Because
of the separability in its kernel,
the solution to the modified
integral equation is obtained explicitly by using the methods from linear algebra. Using that solution, we are able to construct
the Darboux transformation formulas both
at the potential and wavefunction levels.

We recall that
two of the main methods used to solve various inverse problems, namely the Marchenko method \cite{AK2001,CS1989,L1986,M1987} and the
Gel'fand--Levitan method \cite{AK2001,CS1989,GL1955,L1986,M1987}, both involve the use of a linear integral equation.
The Marchenko integral equation can be formulated on the interval $(x,+\infty)$
when the input data set is related to the measurements at $x=+\infty,$ and it can be formulated
on $(-\infty,x)$ when the input data set is related to the measurements at $x=-\infty.$
The Gel'fand--Levitan integral equation can be formulated
on the interval $(0,x)$ when the input data set is related to the measurements at $x=0.$
Thus, a fundamental integral equation arises naturally as
the Marchenko integral equation on $(x,+\infty)$ and on $(-\infty,x)$
and also as the Gel'fand--Levitan integral equation on $(0,x).$
In our generalized method, we use a fundamental integral equation for
each of the unperturbed and perturbed problems rather than
a fundamental integral equation used only for the perturbed problem in the
Marchenko and Gel'fand--Levitan methods.

The generalized method introduced here uses the basic ideas developed in \cite{AV2009} on
the interval $(x,+\infty)$ and in the second author's doctoral thesis \cite{U2014}
on the intervals $(-\infty,x)$ and $(0,x).$
We remark that the proofs and the details of the analysis
on $(x,+\infty)$ and those on $(-\infty,x)$ and $(0,x)$ are not trivially related. In order to emphasize the
unified aspect of our approach, we provide a summary of the relevant results in \cite{AV2009} on $(x,+\infty)$
without any proofs, and we present the results on $(-\infty,x)$ and on $(0,x)$ with some brief proofs.
At various places in our paper, we demonstrate the unified aspect of our approach in the three aforementioned intervals.

Our paper is organized as follows. In Section~\ref{section2} we describe our generalized method
for the Darboux transformation in operator notation, and we present the
corresponding results on the interval $(x,+\infty)$ without any proofs.
In that section, we also provide a comparison between
our generalized method and the standard method in the case of the full-line Schr\"odinger operator
when a bound state is added to the spectrum.
In Section~\ref{section3} our generalized method
is described on the interval $(-\infty,x)$ with some details and brief proofs, and
a comparison is provided between our generalized method and the standard
method when a bound state is added to the spectrum of full-line Schr\"odinger operator.
In Section~\ref{section4} we describe our generalized method
on the interval $(0,x),$ and we also present
a comparison between our generalized method and the standard method
when a bound state is added to the spectrum of the half-line Schr\"odinger operator
with the Dirichlet boundary condition and also with a non-Dirichlet boundary condition.
Finally, in Section~\ref{section5} we illustrate the results presented
in the earlier sections with some explicit examples, demonstrate
how our generalized method works, and clarify some subtle
points in our method.

\section{The generalized method on the interval $(x, +\infty)$}
\label{section2}

In this section we present our generalized method for the Darboux transformation for the linear system $\mathcal L\Psi = \lambda \Psi$
with the help of the fundamental linear integral equations on the interval $(x,+\infty)$ both for the unperturbed and perturbed linear
systems. We describe our method in operator notation, which provides the appropriate
preliminaries in such a way that our generalized approach can be applied
on the interval $(-\infty,x)$ in Section~\ref{section3}
and on $(0,x)$ in Section~\ref{section4}.
At the end of the section we compare our generalized approach with the standard
approach for the full-line Schr\"odinger equation
when a bound state is added to the spectrum.

We summarize our generalized approach for the Darboux transformation on $(x,+\infty)$ as follows.
For the unperturbed problem $\mathcal L\Psi =\lambda \Psi,$ we introduce the fundamental integral equation
\begin{equation}
\label{2.1}
\alpha(x,y) + \omega(x,y)+ \int_x^\infty dz \, \alpha(x,z) \, \omega(z,y)=0, \qquad x<y,
\end{equation}
where $\alpha(x,y)$ is the unknown. We remark that the nonhomogeneous term and the
integral kernel in \eqref{2.1} coincide, and hence \eqref{2.1} can be viewed as a 
Marchenko integral equation \cite{AK2001,AW2006,CS1989,DT1979,L1986,M1987,N1983}
on the interval $(x,+\infty).$
The integral equation \eqref{2.1} usually arises 
by taking the Fourier transform of a relationship involving the scattering data and certain wavefunctions
for the unperturbed problem $\mathcal L\Psi = \lambda \Psi.$ The quantity $\alpha$ is related to the Fourier
transform of a specific wavefunction $\Psi,$ and $\omega$ is related to the Fourier transform of some scattering and spectral data set
$S(\lambda)$ associated with the unperturbed operator $\mathcal L.$ We assume that the integral equation \eqref{2.1}
is uniquely solvable in some function space. We can write \eqref{2.1} in operator form as
\begin{equation}
\label{2.2}
\alpha+\omega+\alpha\,\Omega=0,
\end{equation}
where the integral operator $\Omega$ acts from the right.

Corresponding to the operator $\Omega$ in \eqref{2.2}, let us define the operator $R$ as
\begin{equation}
\label{2.3}
R:=(I+\Omega )^{-1} -I,
\end{equation}
where $I$ is the identity operator. We then have
\begin{equation*}
I+R=(I+\Omega)^{-1}.
\end{equation*}
We refer to $R$ as the resolvent operator for \eqref{2.2}. Using \eqref{2.3} we express the unique
solution $\alpha(x,y)$ to \eqref{2.2} as
\begin{equation}
\label{2.4}
\alpha=-\omega\, (I+R).
\end{equation}
Let us use $r(x;y,z)$ to denote the kernel of the integral operator $R.$ We refer to it as the resolvent kernel.
On the interval $(x,+\infty),$ we note that \eqref{2.4} is equivalent to
\begin{equation}
\label{2.5}
\alpha(x,y)=-\omega(x,y)-\int_x^\infty dz\,\omega(x,z)\,r(x;z,y), \qquad x<y.
\end{equation}
Without much loss of generality, we consider \eqref{2.2} when the integral operator $\Omega$ is $N\times N$
matrix valued and also $J$-selfadjoint in the sense that
\begin{equation}
\label{2.6}
\Omega=J\,\Omega^\dagger J, \quad \omega(y,z)=J\,\omega(z,y)^\dagger J,
\end{equation}
where the dagger denotes the matrix adjoint (complex conjugate and matrix transpose), $N$ is a positive integer,
and $J$ is an $N\times N$ selfadjoint involution, i.e.
\begin{equation}
\label{2.7}
J^{-1}=J, \quad J^\dagger=J.
\end{equation}
For instance, $J$ may be chosen as $I,$ $-I,$ or a diagonal block matrix of the form
\begin{equation}
\label{2.8}
J:=
\begin{bmatrix}
I_j & 0\\
0 & -I_{N-j}
\end{bmatrix},
\end{equation}
where $I_j$ is the $j \times j$ identity matrix for some $1\leq j< N.$ The use of $J$ allows us
to apply our method on a larger class of spectral problems.

In the scalar case, i.e. when $N=1,$ from \eqref{2.6} it follows that $\omega(y,z)$ is real valued and symmetric
in $y$ and $z,$ i.e. $\omega(y,z)=\omega(z,y).$ In the matrix case, i.e. when $N\geq 2,$ it follows that the
diagonal entries $\omega_{jj}(y,z)$ for $1\leq j\leq N$ are real valued and symmetric in $y$ and $z,$ i.e.
$\omega_{jj}(y,z)=\omega_{jj}(z,y).$ In the matrix case, the corresponding off-diagonal entries $\omega_{jk}(y,z)$ and
$\omega_{kj}(y,z)$ are related to each other as $\omega_{kj}(y,z)=\omega_{jk}(z,y)^*$ or
$\omega_{kj}(y,z)=-\omega_{jk}(z,y)^*,$ depending on the negative sign in the involution matrix $J$
appearing in \eqref{2.8}. Note that we use an asterisk to denote complex conjugation.

We assume that the operator $\Omega$ acts on the complex Hilbert space $\mathcal H^2$ of $N \times N$
matrix-valued measurable functions $M:(x,+\infty) \to \mathbb C^{N \times N}$ whose matrix norms belong
to $L^2(x,+\infty).$

In analogy
with the integral equation \eqref{2.2} for the unperturbed problem, we have the fundamental integral equation
associated with the perturbed problem $\tilde{\mathcal L}\, \tilde\Psi = \lambda \, \tilde\Psi,$
and it is given by
\begin{equation}
\label{2.9}
\tilde\alpha(x,y)+\tilde{\omega}(x,y)+\int_x^\infty dz \, \tilde\alpha(x,z)\, \tilde\omega(z,y)=0, \qquad x<y,
\end{equation}
which is represented in operator form as
\begin{equation}
\label{2.10}
\tilde\alpha + \tilde\omega + \tilde\alpha\, \tilde\Omega= 0.
\end{equation}

In the Darboux transformation, the perturbation corresponds to the case where the integral
operators $\tilde\Omega$ and $\Omega$ differ by a finite-rank
operator. Let us use $FG$ to denote that finite-rank perturbation operator and use $f(x)\, g(y)$ to denote the
corresponding kernel. Thus, we have
\begin{equation}
\label{2.11}
\tilde\Omega=\Omega+F G, \quad \tilde\omega(x,y)=\omega(x,y)+f(x)\, g(y).
\end{equation}
Since we deal with $N\times N$ matrix-valued quantities, when $N\geq2$ the operators
$F$ and $G$ do not necessarily commute, and hence in general $f(x)g(y) \neq g(y)f(x).$
We remark that the formulation of the perturbation as in \eqref{2.11} is valid
regardless whether discrete eigenvalues are added to or removed from the spectrum.
In fact, those two cases can be handled in the same manner by simply changing the sign of
the quantity specified by either of $f(x)$ or $g(y).$

The goal in the Darboux transformation is to determine the perturbed potential $\tilde u$ and the perturbed
wavefunction $\tilde\Psi$ when we know the unperturbed quantities $u$ and $\Psi$ as well as the perturbation
$\tilde\Omega-\Omega.$ Our generalized approach to obtain the Darboux transformation consists of the following steps:

\begin{enumerate}
\item[\text{\rm(a)}] Knowing the solution $\alpha(x,y)$ to the unperturbed integral equation \eqref{2.1}, and also
knowing the perturbation quantities $f(x)$ and $g(y)$, we construct the intermediate quantities $n(x)$ and $q(y)$ as
\begin{equation}
\label{2.12}
n(x):=f(x)+\int_x^\infty dz\, \alpha(x,z)\,f(z),
\end{equation}
\begin{equation}
\label{2.13}
q(y):=g(y)+\int_y^\infty dz\, g(z)\,J\alpha(y,z)^\dagger J,
\end{equation}
where we recall that $J$ is the involution matrix appearing in \eqref{2.7}.

\item[\text{\rm(b)}] Next, by using $\alpha(x,y)$ and $q(y)$ we introduce the auxiliary quantity $\tilde g(x,y)$ defined as
\begin{equation}
\label{2.14}
\tilde g(x,y):=q(y)+\int_x^y dz\, q(z)\, \alpha(z,y).
\end{equation}

\item[\text{\rm(c)}] Then, in terms of $f(x)$ and $\tilde g(x,y)$
we introduce the $N\times N$ matrix-valued quantity $\Gamma(x)$  given by
\begin{equation}\label{2.15}
\Gamma(x):=I+\int_x^\infty dz\,\tilde g(x,z)\,f(z),
\end{equation}
where $I$ is the $N\times N$ identity matrix.

\item[\text{\rm(d)}] We then show that the solution $\tilde\alpha(x,y)$ to the perturbed integral equation \eqref{2.9} can be
expressed in terms of the already known quantities $\alpha(x,y),$ $n(x),$
$\tilde g(x,y),$ $\Gamma(x),$ and we have
\begin{equation}\label{2.16}
\tilde\alpha(x,y)=\alpha(x,y)-n(x)\,\Gamma(x)^{-1}\,\tilde g(x,y), \qquad x<y.
\end{equation}

\item[\text{\rm(e)}] The change $\tilde u(x)-u(x)$ in the potential is obtained from the quantity $\tilde\alpha(x,y)-\alpha(x,y)$
in the limit $y\to x^+,$ and we use $\tilde\alpha(x,x)-\alpha(x,x)$ to denote that limit. From \eqref{2.14} and \eqref{2.16}, we see that
$\tilde u(x)-u(x)$ is determined by the auxiliary quantities $n(x),$ $q(x),$ and $\Gamma(x)$ as
\begin{equation}\label{2.17}
\tilde\alpha(x,x)-\alpha(x,x)=-n(x)\,\Gamma(x)^{-1}\,q(x).
\end{equation}
The specific process of obtaining
$\tilde u(x)-u(x)$ from $\tilde\alpha(x,x)-\alpha(x,x)$ depends on the particular unperturbed linear problem
$\mathcal L\Psi = \lambda \Psi,$ but that specific process is usually straightforward. The resulting equality
expressing $\tilde u(x)-u(x)$ in terms of $\tilde\alpha(x,x)-\alpha(x,x)$ constitutes the Darboux transformation at the potential level.

\item[\text{\rm(f)}] Let us recall that the unperturbed and perturbed wavefunctions $\Psi$ and $\tilde\Psi$  are 
usually related to the quantities $\alpha(x,y)$ and $\tilde\alpha(x,y),$ respectively, via a Fourier transformation.
Under a Fourier transformation, the change $\tilde\Psi(\lambda,x)-\Psi(\lambda,x)$ in the wavefunction
is readily expressed in terms of the already constructed quantity $\tilde\alpha(x,y)-\alpha(x,y).$ The resulting equation
corresponds to the Darboux transformation at the wavefunction level. 
We mention that the specific process of obtaining
$\tilde\Psi(\lambda,x)-\Psi(\lambda,x)$ from $\tilde\alpha(x,y)-\alpha(x,y)$ depends on the particular unperturbed linear problem
$\mathcal L\Psi = \lambda \Psi,$ but that specific process is usually straightforward. 

\end{enumerate}

We have outlined the above steps for our generalized approach to the Darboux transformation on the interval $(x,+\infty)$ 
so that they are readily
applicable also on $(-\infty,x)$ and on $(0,x).$ In fact, the steps listed above can be described in operator form without
any specific reference to any of the intervals $(x,+\infty),\,(-\infty,x),$ and $(0,x).$ Thus, we refer to our
approach as a generalized method for the Darboux transformation.

Although our generalized approach can be described in a unified way in operator notation for the three intervals $(x,+\infty),\,(-\infty,x),$ and $(0,x),$
it is still relevant and important to present some proofs and details on each of those three intervals separately.
This is because those proofs and details are not necessarily trivial extensions from any one of those three intervals.
Furthermore, as already indicated, the choice of a relevant specific wavefunction usually depends on the
particular interval $(x,+\infty),$ $(-\infty,x),$ or $(0,x).$

Let us remark that the quantity $\tilde g(x,y)$ defined in step (b) in our generalized Darboux method can equivalently
be evaluated as
\begin{equation}
\label{2.18}
\tilde g(x,y)=g(y)+\int_x^\infty dz\, g(z)\, r(x;z,y),
\end{equation}
where we recall that $g(y)$ is the perturbation quantity appearing in \eqref{2.11} and \eqref{2.13}, and $r(x;z,y)$ is
the resolvent kernel for the operator $R$ given in \eqref{2.3}. Although \eqref{2.14} and \eqref{2.18}
are equivalent, in order to use \eqref{2.18} we first need to construct the
resolvent kernel $r(x;z,y).$ In the next theorem we show that $r(x;z,y)$ can explicitly be evaluated in terms of the unique solution
$\alpha(x,y)$ to the integral equation \eqref{2.1}.
We omit the proof and refer the reader to Proposition~2.1 and Theorem~2.2
of \cite{AV2009} for that proof.
We remark that the proof itself is
not trivial, but the resulting formula is relevant and important.

\begin{theorem}
\label{theorem2.1}
Assume that \eqref{2.2} is uniquely solvable in the aforementioned Hilbert space $\mathcal H^2$ and that the operator $\Omega$ satisfies \eqref{2.6}.
Then, the resolvent operator $R$ appearing in \eqref{2.3} and the corresponding kernel $r(x;y,z)$ satisfy
\begin{equation}
\label{2.19}
R=J\,R^\dagger J, \quad r(x;y,z)=J\,r(x;z,y)^\dagger J,
\end{equation}
where we recall that $J$ is the involution matrix appearing in \eqref{2.6}. Furthermore, $r(x;y,z)$ is expressed explicitly
in terms of the solution $\alpha(x,y)$ to \eqref{2.1} as
\begin{equation}
\label{2.20}
r(x;y,z)=\begin{cases}\alpha(y,z)+\ds\int_x^y ds\,J\,\alpha(s,y)^\dagger J\,\alpha(s,z), \qquad x<y<z,\\
J\,\alpha(z,y)^\dagger J+\ds\int_x^z ds\,J\,\alpha(s,y)^\dagger J\,\alpha(s,z),\qquad x<z<y.\end{cases}
\end{equation}
\end{theorem}

The remark made after \eqref{2.8} on the $J$-selfadjointness of the operator kernel $\omega(y,z)$ also applies to the
resolvent kernel $r(x;z,y)$ appearing in the second equality in \eqref{2.19}. Comparing the second equalities
of \eqref{2.6} and \eqref{2.19}, we have the following observations on $r(x;z,y).$ In the scalar case, i.e. when $N=1,$
the quantity $r(x;z,y)$ is real valued and satisfies the symmetry property $r(x;z,y)=r(x;y,z).$ In the matrix case, i.e.
when $N\geq 2,$ the diagonal entries $r_{jj}(x;z,y)$ are real and symmetric in $z$ and $y,$ i.e. we have $r_{jj}(x;z,y)=r_{jj}(x;y,z)$
for $1\leq j\leq N.$ On the other hand, the corresponding off-diagonal entries $r_{jk}(x;z,y)$ and $r_{kj}(x;z,y)$
satisfy $r_{kj}(x;z,y)=r_{jk}(x;y,z)^*$ or $r_{kj}(x;z,y)=-r_{jk}(x;y,z)^*,$ depending on the
appearance of the negative sign in the involution matrix $J.$

Let us also remark that, for a given pair $\alpha(x,y)$ and $\omega(x,y),$ we may have more than one function that can
be substituted for $r(x;z,y)$ in \eqref{2.5} so that \eqref{2.5} is satisfied. However, not all such functions act as a resolvent
kernel for the operator $\Omega$ appearing in \eqref{2.3} and \eqref{2.6}, but only one of them is the resolvent
kernel for $\Omega.$ The correct function to be substituted for $r(x;z,y)$ in \eqref{2.5} must be the resolvent kernel of
$\Omega,$ and that resolvent kernel is uniquely determined by the solution $\alpha(x,y)$ to \eqref{2.1}, as indicated in
Theorem~\ref{theorem2.1}. The elaboration on this issue is provided later
in Example~\ref{example5.2}.

The following theorem is a key result in the implementation of the steps in our generalized approach to the Darboux
transformation. It is used to construct $\tilde\alpha(x,y)$ in terms of the unperturbed quantities and the perturbation.
We present it without a proof, and we refer the reader to \cite{AV2009} for its proof.

\begin{theorem}
\label{theorem2.2}
Under the finite-rank perturbation $\tilde\Omega-\Omega$ given in \eqref{2.11}, the perturbed integral equation
\eqref{2.9} involving $\tilde\alpha(x,y)$ and $\tilde\omega(x,y)$ can be transformed into an integral equation with
a separable kernel. In fact, that transformed integral equation has the kernel $f(y)\,\tilde g(x,z),$ where $f$
and $\tilde g$ are the quantities appearing in \eqref{2.11} and \eqref{2.14}, respectively. That transformed
integral equation is given by
\begin{equation}
\label{2.21}
\tilde\alpha(I+F\,\tilde G)=\alpha-f\tilde g,
\end{equation}
where $F$ is the operator appearing in the first equality of \eqref{2.11} and
$\tilde G$ is the operator defined
as $\tilde G:=G(I+R),$
with $R$ and $G$ being the
operators appearing in \eqref{2.3} and \eqref{2.11}, respectively.
The kernel $f(y)\,\tilde g(x,z)$ of the integral operator $F\tilde G$ is separable in $y$ and $z,$
where the appearance of the parameter $x$ does not affect the separability. Consequently,
the transformed integral equation \eqref{2.21} is explicitly solvable by the methods of linear algebra, and
the solution $\tilde\alpha(x,y)$ to \eqref{2.9} is given by
\begin{equation*}
\tilde\alpha(x,y)=\alpha(x,y)-n(x)\left[I+\ds\int_x^\infty ds\, \tilde g(x,s)\, f(s)
\right]^{-1}\tilde g(x,y),
\end{equation*}
where $\alpha(x,y)$ and $n(x)$ are the quantities appearing in
\eqref{2.1} and \eqref{2.12}, respectively.
\end{theorem}

Having described our generalized approach on $(x,+\infty)$ to the Darboux transformation, let us briefly
illustrate the difference between the standard approach and our generalized approach on a specific linear system, namely for the scalar
Schr\"odinger equation on the full line 
\begin{equation}
\label{2.22}
-\ds\frac{d^2\psi(k,x)}{dx^2}+u(x)\,\psi(k,x)=k^2\,\psi(k,x),\qquad -\infty<x<+\infty,
\end{equation}
where the
linear operator $\mathcal L$ is given by
\begin{equation*}
\mathcal L=-\ds\frac{d^2}{dx^2}+u(x),
\end{equation*}
the appropriate wavefunction to use is the left Jost solution $f_l(k,x)$ 
satisfying 
the spacial asymptotics
\begin{equation*}
f_l(k,x)=e^{ikx}[1+o(1)],\qquad x\to +\infty,
\end{equation*}
with $k$ being the spectral parameter related to $\lambda$ as $\lambda=k^2.$ 
In the standard approach, one must use not only the wavefunction $f_l(k,x)$
but also the right Jost solution
$f_r(k,x)$ to \eqref{2.22} with the spacial asymptotics
\begin{equation}
\label{2.23}
f_r(k,x)=e^{-ikx}[1+o(1)],\qquad x\to -\infty.
\end{equation}
If a bound state at $k=\text{i}\kappa$
with the dependency constant $\gamma$ is to be added to the discrete spectrum, that bound state with energy $-\kappa^2$
must be added below the already existing bound-state energies.
This is a limitation in the standard method and it is needed to ensure that $\eta(x)>0,$
where $\eta(x)$ is the quantity defined as
\begin{equation}
\label{2.24}
\eta(x):=f_l(i\kappa,x)+\gamma f_r(i\kappa,x).
\end{equation}
We then obtain \cite{AK2001,CS1989,DT1979,MS1991} the perturbed 
Schr\"odinger equation
\begin{equation*}
-\ds\frac{d^2\tilde\psi(k,x)}{dx^2}+\tilde u(x)\,\tilde\psi(k,x)=k^2\,\tilde\psi(k,x),\qquad -\infty<x<+\infty,
\end{equation*}
with the
potential $\tilde u(x)$ specified as
\begin{equation}
\label{2.25}
\tilde u(x)=u(x)-2\ds\,\frac{d^2\ln\left( \eta(x)\right)}{dx^2},
\end{equation}
and the perturbed Jost solutions $\tilde f_l(k,x)$ and $\tilde f_r(k,x)$ as
\begin{equation}
\label{2.26}
\tilde f_l(k,x)=\ds\frac{1}{i(k+i\kappa)}\left[f'_l(k,x)-\ds\frac{\eta'(x)}{\eta(x)}\,f_l(k,x)\right],
\end{equation}
\begin{equation}
\label{2.27}
\tilde f_r(k,x)=\ds\frac{i}{k+i\kappa}\left[f'_r(k,x)-\ds\frac{\eta'(x)}{\eta(x)}\,f_r(k,x)\right],
\end{equation}
where the prime denotes the $x$-derivative and 
 the dependency constant $\gamma$ is given by
\begin{equation*}
\gamma:=\ds\frac{\tilde f_l(i\kappa,x)}{\tilde f_r(i\kappa,x)}.
\end{equation*}
We remark that the dependency constant $\gamma$ can be expressed \cite{AK2001} as
\begin{equation*}
\gamma=\ds\frac{2\kappa\,T(i\kappa)}{c_l^2},
\end{equation*}
where $T(k)$ is the transmission coefficient corresponding to the unperturbed potential $u(x)$ in \eqref{2.22}.
The norming constant $c_l$ is related to the perturbed left Jost solution $\tilde f_l(k,x)$ as
\begin{equation*}
c_l:=\left[ \ds\int_{-\infty}^\infty dx\,\tilde f_l(i\kappa,x)^2   \right]^{-1/2}.
\end{equation*}

Let us now briefly present our generalized approach and make a contrast with the standard approach.
In our generalized method, the second equality in \eqref{2.11} is given by
\begin{equation*}
\tilde\omega(x,y)-\omega(x,y)=c_l^2\,e^{-\kappa(x+y)},
\end{equation*}
and hence the perturbation quantities $f(x)$ and $g(y)$ can be chosen as
\begin{equation*}
f(x)=c_l\,e^{-\kappa x},\quad g(y)=c_l\,e^{-\kappa y}.
\end{equation*}
The corresponding intermediate quantities $n(x),$ $q(y),$ $\Gamma(x),$ and $\tilde g(x,y)$ appearing 
in \eqref{2.12}, \eqref{2.13}, \eqref{2.15}, and \eqref{2.14}, respectively, are evaluated as
\begin{equation}
\label{2.28}
n(x)=c_l\,f_l(i\kappa,x),\quad q(y)=c_l\,f_l(i\kappa,y),\quad 
\Gamma(x)=1+c_l^2 \ds\int_x^\infty dz\,f_l(i\kappa,z)^2,
\end{equation}
\begin{equation*}
\tilde g(x,y)=c_l\,f_l(i\kappa,y)+c_l \ds\int_x^\infty dz\,f_l(i\kappa,z)\,\alpha(z,y),
\end{equation*}
and we have
\begin{equation*}
\alpha(x,y)=\ds\frac{1}{2\pi}\ds\int_{-\infty}^\infty dk\left[f_l(k,x)-e^{ikx}
\right]e^{-iky},\quad 
\tilde\alpha(x,y)=\ds\frac{1}{2\pi}\ds\int_{-\infty}^\infty dk\left[\tilde f_l(k,x)-e^{ikx}
\right]e^{-iky},
\end{equation*}
\begin{equation*}
f_l(k,x)=e^{ikx}+  \ds\int_x^\infty dy\,    \alpha(x,y)\,e^{iky},\quad 
\tilde f_l(k,x)=e^{ikx}+  \ds\int_x^\infty dy\,   \tilde\alpha(x,y)\,e^{iky}.
\end{equation*}
Then, the counterpart of \eqref{2.25} is obtained as
\begin{equation}
\label{2.29}
\tilde u(x)-u(x)=-2\ds\,\frac{d^2\ln \left(\Gamma(x)\right)}{dx^2},
\end{equation}
where $\Gamma(x)$ is the quantity appearing in the last equality of \eqref{2.28},
and the counterpart of \eqref{2.26} is given by
\begin{equation}
\label{2.30}
\tilde f_l(k,x)-f_l(k,x)=-\ds\frac{c_l^2\,f_l(i\kappa,x)  \int_x^\infty dy\, f_l(k,y)\,
f_l(i\kappa,y)}{1+c_l^2  \int_x^\infty dz\, f_l(i\kappa,z)^2}.
\end{equation}
With the help of
\begin{equation*}
-f''_l(k,x)+u(x)\,f_l(k,x)=k^2\,f_l(k,x),\quad
-f''_l(i\kappa,x)+u(x)\,f_l(i\kappa,x)=-\kappa^2 \,f_l(i\kappa,x),
\end{equation*}
we can express the integral appearing in the numerator in \eqref{2.30} as
\begin{equation}
\label{2.31}
\int_x^\infty dy\, f_l(k,y)\,
f_l(i\kappa,y)=\ds\frac{1}{k^2+\kappa^2}\left[f_l'(k,x)\,f_l(i\kappa,x)-f_l(k,x)\,f_l'(i\kappa,x)
\right],
\end{equation}
and hence \eqref{2.30} can also be written in the equivalent form as
\begin{equation*}
\tilde f_l(k,x)=\left[1+\ds\frac{c_l^2\,f_l(i\kappa,x) \,f'_l(i\kappa,x)}
{(k^2+\kappa^2)\left[1+c_l^2  \int_x^\infty dz\, f_l(i\kappa,z)^2\right]}
\right] f_l(k,x)
-\ds\frac{c_l^2\,f_l(i\kappa,x)^2 \,f'_l(k,x)}
{(k^2+\kappa^2)\left[1+c_l^2  \int_x^\infty dz\, f_l(i\kappa,z)^2\right]}.
\end{equation*}

We remark that our generalized formulas \eqref{2.29} and \eqref{2.30} use
only the relevant wavefunction $f_l(k,x)$ and the bound-state
parameters $\kappa$ and $c_l$ whereas
the standard formulas \eqref{2.25} and \eqref{2.26} not only use
the relevant wavefunction $f_l(k,x)$ but also the information on the unperturbed right Jost solution $f_r(k,x)$ and the unperturbed transmission
coefficient $T(k).$
Furthermore, contrary to the standard method, in our generalized
method there is no restriction that the bound state at $k=i\kappa$ must be added to the discrete spectrum below
the already existing bound-state energy levels.
This is because the quantity $\Gamma(x)$ defined in \eqref{2.28} already satisfies $\Gamma(x)>0.$

\section{The generalized method on the interval $(-\infty,x)$}
\label{section3}

Our primary goal in this section is to develop our generalized method on the interval $(-\infty,x)$ 
for the Darboux transformation for the linear system $\mathcal L\Psi = \lambda \Psi,$ and this is done by introducing
and using
the fundamental linear integral equations for the corresponding unperturbed and perturbed linear
systems.
At the end of the section we compare our generalized method with the standard
method for the full-line Schr\"odinger equation
when a bound state is added to the spectrum.

Let us recall that we are interested in obtaining the Darboux transformation formulas at the potential and
wavefunction levels when the fundamental integral equation for the unperturbed system is given by
\begin{equation}
\label{3.1}
\alpha(x,y) + \omega(x,y)+ \int_{-\infty}^x dz \, \alpha(x,z) \, \omega(z,y)=0, \qquad y<x,
\end{equation}
where we emphasize that the integration is on the interval $(-\infty,x)$ rather than $(x,+\infty)$ used in \eqref{2.1}. Note that both
\eqref{2.1} and \eqref{3.1} correspond to \eqref{2.2} written in operator notation, which is consistent with
the fact that we refer to our method as a generalized method for the Darboux transformation. In describing our
method on $(-\infty,x),$ we follow the description outlined in
Section~\ref{section2} on $(x,+\infty),$ and we provide the motivation for the steps and some brief details wherever appropriate.

Our generalized method is applied on the interval $(-\infty,x)$ as follows. As indicated in \eqref{2.4} and \eqref{2.5},
we first express the unique solution $\alpha(x,y)$ to \eqref{3.1} as
\begin{equation}
\label{3.2}
\alpha(x,y)=-\omega(x,y)-\int_{-\infty}^x dz\,\omega(x,z)\,r(x;z,y), \qquad y<x,
\end{equation}
where $r(x;z,y)$ corresponds to the kernel of the resolvent $R$ appearing in \eqref{2.4}. The resolvent kernel on the
interval $(x,+\infty)$ is explicitly expressed in terms of the unique solution $\alpha(x,y)$ to \eqref{2.1}, and that result is
presented in Theorem~\ref{theorem2.1}. In a similar way, we express the resolvent kernel $r(x;z,y)$ on the interval $(-\infty,x)$
in terms of the unique solution $\alpha(x,y)$ to \eqref{3.1}, and our result presented in the next theorem. We remark
that the proof of the next theorem is not a trivial extension of the proof of Theorem~\ref{theorem2.1}. 
We present a relatively brief proof of our next theorem on $(-\infty,x),$ and
we refer the reader to Proposition~3.1.1 and Theorem~3.1.2 of \cite{U2014} for the details of the proof. 

\begin{theorem}
\label{theorem3.1}
Assume that \eqref{3.1} is uniquely solvable for $\alpha(x,y)$ in the aforementioned Hilbert space $\mathcal H^2.$ Also suppose that the operator $\Omega$ with
the kernel $\omega(x,y)$ appearing in \eqref{3.1} satisfies \eqref{2.6} and that $\Omega$ is related to the resolvent
$R$ with the kernel $r(x;z,y)$ as in \eqref{2.4}. We then have the following:
\begin{enumerate}
\item[\text{\rm(a)}] The resolvent $R$ and the kernel $r(x;z,y)$ satisfy \eqref{2.19}.

\item[\text{\rm(b)}] The kernel $r(x;z,y)$ is explicitly expressed in terms of the solution $\alpha(x,y)$ to \eqref{3.1} as
\begin{equation}
\label{3.3}
r(x;z,y)=\begin{cases}\alpha(z,y)+\ds\int_z^x ds\,J\,\alpha(s,z)^\dagger J\,\alpha(s,y), \qquad y<z<x,\\
J\,\alpha(y,z)^\dagger J+\ds\int_y^x ds\,J\,\alpha(s,z)^\dagger J\,\alpha(s,y),\qquad z<y<x,\end{cases}
\end{equation}
where we recall that $J$ is the involution matrix appearing in \eqref{2.6}.
\end{enumerate}
\end{theorem}

\begin{proof}
We only provide an outline
for the proof of (a) and refer the reader to Proposition~3.1.1 of \cite{U2014} for the details. With
the help of the first equality in \eqref{2.3}, we obtain the two operator equations
\begin{equation}
\label{3.4}
R+\Omega+R\,\Omega=0,
\end{equation}
\begin{equation}
\label{3.5}
R+\Omega+\Omega\,R=0.
\end{equation}
By taking the adjoint and then applying $J$ on both sides in \eqref{3.5}, we get
\begin{equation}
\label{3.6}
JR^\dagger J+J\,\Omega^\dagger J+(J\,\Omega^\dagger J)(JR^\dagger J)=0.
\end{equation}
Using the first equality of \eqref{2.6}, we write \eqref{3.6} as
\begin{equation}
\label{3.7}
J R^\dagger J+\Omega+\Omega\,(J R^\dagger J)=0.
\end{equation}
Since \eqref{3.1} is assumed to be uniquely solvable in $\mathcal H^2,$ by comparing \eqref{3.7} with \eqref{3.5},
we see that \eqref{2.19} holds on the interval $(-\infty,x).$ Thus, the proof of (a) is complete. We now turn to the
proof of (b). Since \eqref{3.1} is assumed to be uniquely solvable in $\mathcal H^2,$ the corresponding operator
equation \eqref{2.2} is also uniquely solvable on $(-\infty,x),$ and hence the solution $R$ to \eqref{3.4} is unique.
Thus, it suffices to prove that the quantity $r(x;z,y)$ expressed in \eqref{3.3} satisfies the integral equations
\begin{equation}
\label{3.8}
r(x;z,y)+ \omega(z,y)+ \int_{-\infty}^x ds\, r(x;z,s)\, \omega(s,y)=0, \qquad y<z<x,
\end{equation}
\begin{equation}
\label{3.9}
r(x;z,y)+ \omega(z,y)+ \int_{-\infty}^x ds\, r(x;z,s)\, \omega(s,y)=0, \qquad z<y<x.
\end{equation}
We only give the proof for \eqref{3.8}. The proof for \eqref{3.9} is much more challenging, and we refer the
reader to Theorem~3.1.2 of \cite{U2014} for that
proof. For the proof of \eqref{3.8} we proceed as follows. Using $\int_{-\infty}^x = \int_{-\infty}^z+\int_z^x$  in
\eqref{3.8}, we write the left-hand side of \eqref{3.8} as
\begin{equation}
\label{3.10}
r(x;z,y)+\omega(z,y)+\int_{-\infty}^z ds\, r(x;z,s)\, \omega(s,y)+\int_z^x ds\, r(x;z,s)\, \omega(s,y).
\end{equation}
We use the first and second lines of \eqref{3.3} in the integrals $\int_{-\infty}^x$ and $\int_z^x$ in \eqref{3.10}, respectively.
Then, we find that the left-hand side of \eqref{3.8} can be written as
\begin{equation}
\label{3.11}
\begin{split}
\alpha(z,y)&+\int_z^x ds\, J\,\alpha(s,z)^\dagger J\, \alpha(s,y)+ \omega(z,y)\\
&+\int_{-\infty}^z ds\left[\alpha(z,s)+\int_z^x dt\, J\,\alpha(t,z)^\dagger J\, \alpha(t,s) \right]\, \omega(s,y)\\
&+\int_z^x ds \left[J\,\alpha(s,z)^\dagger J +\int_s^x dt\, J\,\alpha(t,z)^\dagger J\, \alpha(t,s) \right]\, \omega(s,y).
\end{split}
\end{equation}
Expanding \eqref{3.11}, we write it in the equivalent form as
\begin{equation}
\label{3.12}
b_1+b_2+b_3+b_4,
\end{equation}
where we have defined
\begin{equation*}
b_1:=\,\alpha(z,y)+\omega(z,y)+\int_{-\infty}^z ds\, \alpha(z,s)\, \omega(s,y),
\end{equation*}
\begin{equation}
\label{3.13}
b_2:=\int_z^x dt\, J\,\alpha(t,z)^\dagger  J\, \alpha(t,y)+\int_z^x dt\, J\,\alpha(t,z)^\dagger J\, \omega(t,y),
\end{equation}
\begin{equation*}
b_3:=\int_{-\infty}^z ds\int_z^x dt\, J\,\alpha(t,z)^\dagger J\, \alpha(t,s)\, \omega(s,y),
\end{equation*}
\begin{equation*}
b_4:=\int_z^x ds \int_s^x dt\, J\,\alpha(t,z)^\dagger J\, \alpha(t,s)\, \omega(s,y).
\end{equation*}
We have $b_1=0$ as a result of \eqref{3.1}. The order of the iterated integrals in $b_3$ and $b_4$ can be
changed to $\int_z^x dt  \int_{-\infty}^z ds$ and $\int_z^x dt \int_z^t ds,$ respectively. Using
$\int_{-\infty}^z+\int_z^t = \int_{-\infty}^t,$ we obtain
\begin{equation}
\label{3.14}
b_3+b_4 =\int_z^x dt\int_{-\infty}^t ds\, J\,\alpha(t,z)^\dagger J\, \alpha(t,s)\, \omega(s,y).
\end{equation}
Next, from \eqref{3.13} and \eqref{3.14}, we get
\begin{equation}
\label{3.15}
b_2+b_3+b_4=\int_z^x dt\, J\,\alpha(t,z)^\dagger J\, \left[\alpha(t,y)+\omega(t,y)+\int_{-\infty}^t ds\, \alpha(t,s)\,
\omega(s,y) \right].
\end{equation}
The integral equation \eqref{3.1} implies that the quantity inside the brackets in \eqref{3.15} vanishes when $t>y,$ and
hence we have $b_2+b_3+b_4=0.$ Since we already have $b_1=0,$ we see that the quantity in \eqref{3.12}
also vanishes and \eqref{3.8} holds. Thus, the proof of (b) is complete.
\end{proof}

Let us remark that, for a given pair $\alpha(x,y)$ and $\omega(x,y),$ the integral equation \eqref{3.2} may also be satisfied if we
substitute other functions for $r(x;z,y)$ instead of the corresponding resolvent kernel, which is uniquely
determined by the solution $\alpha(x,y)$ to \eqref{3.1}. However, those other functions do not satisfy
\eqref{3.8}, and this is later illustrated in Example~\ref{example5.2}. Thus, in \eqref{3.2} the correct function
$r(x;z,y)$ to be used must be the unique resolvent kernel given in \eqref{3.8}.

In the next theorem, on the interval $(-\infty,x)$ we consider the fundamental integral equation \eqref{2.10}
for the perturbed system, i.e. we analyze the linear integral equation
\begin{equation}
\label{3.16}
\tilde\alpha(x,y) + \tilde\omega(x,y)+ \int_{-\infty}^x dz \, \tilde\alpha(x,z) \, \tilde\omega(z,y)=0, \qquad y<x,
\end{equation}
which is the counterpart of \eqref{2.9} given on the interval $(x,+\infty).$ We show that \eqref{3.16} can be transformed into
another linear integral equation with a separable kernel.

\begin{theorem}
\label{theorem3.2}
Assume that \eqref{3.1} is uniquely solvable for $\alpha(x,y)$ in $\mathcal H^2$ and that the operator
$\Omega$ satisfies \eqref{2.6}. Let $R$ be the corresponding resolvent operator as in \eqref{2.4} with the
kernel $r(x;z,y).$ Furthermore, assume that $\tilde\Omega-\Omega$ corresponds to the finite-rank perturbation given in
\eqref{2.11}. Then, we have the following:
\begin{enumerate}
\item[\text{\rm(a)}] The integral equation \eqref{3.16} is transformed into the integral equation
\begin{equation}
\label{3.17}
\tilde\alpha(I+F\,\tilde G)=\alpha-f\tilde g,
\end{equation}
whose kernel $f(y)\,\tilde g(x,z)$ is separable in $y$ and $z.$ Consequently, \eqref{3.17} is
explicitly solvable by using the methods of linear algebra, and the solution $\tilde\alpha(x,y)$ to
\eqref{3.16} is expressed as in \eqref{2.16}, where the relevant quantities are now those defined on the interval
$(-\infty,x)$ rather than on $(x,+\infty).$ In particular, $f(x)$ and $g(y)$ are the quantities appearing in \eqref{2.11}, but $n(x)$
now is the analog of the quantity appearing in \eqref{2.12} and is given by
\begin{equation}
\label{3.18}
n(x):=f(x)+\int_{-\infty}^x dz\, \alpha(x,z)\,f(z),
\end{equation}
the quantity $q(y)$ is the analog of the quantity in \eqref{2.13} and is now given by
\begin{equation}
\label{3.19}
q(y):=g(y)+\int_{-\infty}^y dz\, g(z)\,J\,\alpha(y,z)^\dagger J,
\end{equation}
the quantity $\tilde g(x,y)$ is the analog of the quantity in \eqref{2.14} and is now given by
\begin{equation}
\label{3.20}
\tilde g(x,y):=g(y)+\int_{-\infty}^x dz\, g(z)\, r(x;z,y),
\end{equation}
and the quantity $\Gamma(x)$ is the analog of the quantity in \eqref{2.15} and is now given by
\begin{equation}
\label{3.21}
\Gamma(x):=I+\int_{-\infty}^x dz\,\tilde g(x,z)\,f(z).
\end{equation}
\item[\text{\rm(b)}] The solution $\tilde\alpha(x,y)$ to \eqref{3.16} expressed as in \eqref{2.16} can also be written as
\begin{equation}
\label{3.22}
\tilde\alpha(x,y)=\alpha(x,y)-n(x)\left[ I+\int_{-\infty}^x ds\,\tilde g(x,s)\,f(s) \right]^{-1}\tilde g(x,y), \qquad y<x.
\end{equation}
\item[\text{\rm(c)}] The quantity $\tilde g(x,y)$ defined in \eqref{3.20} can be expressed in terms of the
quantities $\alpha(x,y)$ and $q(y)$ appearing in \eqref{3.1} and \eqref{3.19}, respectively, and we have
\begin{equation}
\label{3.23}
\tilde g(x,y)=q(y)+\int_y^x dz\, q(z)\, \alpha(z,y),
\end{equation}
which is the analog of \eqref{2.14} but expressed on the interval $(-\infty,x).$
\end{enumerate}
\end{theorem}

\begin{proof}
We present the proof in operator notation in order to emphasize the unified aspect of our method for
the Darboux transformation. Using \eqref{2.11} in the version of \eqref{2.10} for the interval $(-\infty,x),$ we get
\begin{equation*}
\omega+f\,g+\tilde\alpha(I+\Omega+F\,G)=0,
\end{equation*}
which yields
\begin{equation}
\label{3.24}
\tilde\alpha(I+\Omega+F\,G)=-\omega-f\,g.
\end{equation}
By applying on \eqref{3.24} from the right with the operator $(I+R)$ appearing in \eqref{2.3}, we obtain
\begin{equation}
\label{3.25}
\tilde\alpha(I+\Omega)(I+R)+\tilde\alpha\,F\,G\,(I+R)=-\omega(I+R)-f\,g\,(I+R).
\end{equation}
Using \eqref{2.4} and the first equality of \eqref{2.3} in \eqref{3.25}, we have
\begin{equation}
\label{3.26}
\tilde\alpha\left[I+FG\,(I+R)\right]=\alpha-fg\,(I+R).
\end{equation}
Next, we introduce the operator $\tilde G$ and its kernel $\tilde g(x,y)$ by letting
\begin{equation}
\label{3.27}
\tilde G:=G\,(I+R), \quad \tilde g(x,y):=g(y)+\int_{-\infty}^x dz\, g(z)\, r(x;z,y),
\end{equation}
where $r(x;z,y)$ is the kernel expressed in \eqref{3.3}. Note that $f(y)\,\tilde g(x,z)$ corresponds to the
kernel of the operator $F\,\tilde G,$ which is a separable kernel in $y$ and $z.$ Using the first equality of \eqref{3.27} in \eqref{3.26}, we obtain
\eqref{3.17}, which is our key integral equation with a separable kernel. Writing \eqref{3.17} as
\begin{equation}
\label{3.28}
(\tilde\alpha-\alpha)+f\tilde g+\tilde\alpha F\tilde G=0,
\end{equation}
we observe that we can express $\tilde\alpha-\alpha$ as a multiple of $\tilde g$ as
\begin{equation}
\label{3.29}
\tilde\alpha(x,y)-\alpha(x,y)=p(x)\,\tilde g(x,y),
\end{equation}
where $p(x)$ is to be determined. Let us write \eqref{3.29} as
\begin{equation}
\label{3.30}
\tilde\alpha(x,y)=\alpha(x,y)+p(x)\,\tilde g(x,y).
\end{equation}
Using \eqref{3.30} in \eqref{3.28}, we get
\begin{equation*}
(p+f+\alpha\,F+p\,\tilde g\,F)\,\tilde G=0,
\end{equation*}
or equivalently we have
\begin{equation}
\label{3.31}
p\,(I+\tilde g\,F)=-(f+\alpha\,F).
\end{equation}
From \eqref{3.31} we recover $p(x)$ as
\begin{equation}
\label{3.32}
p=-(f+\alpha\,F)(I+\tilde g\,F)^{-1}.
\end{equation}
Comparing the right-hand side of \eqref{3.32} with \eqref{3.18} and \eqref{3.21}, we see that we can write
\eqref{3.32} as
\begin{equation}
\label{3.33}
p(x)=-n(x)\,\Gamma(x)^{-1}.
\end{equation}
Using \eqref{3.33} in \eqref{3.30} we obtain
\begin{equation}
\label{3.34}
\tilde\alpha(x,y)=\alpha(x,y)-n(x)\,\Gamma(x)^{-1}\,\tilde g(x,y), \qquad y<x,
\end{equation}
on the interval $(-\infty,x).$ Hence, the proof of (a) is complete. The proof of (b) directly follows by using
\eqref{3.21} in \eqref{3.34}. Let us now prove (c). Using $\int_{-\infty}^x = \int_{-\infty}^y+\int_y^x$
on the right-hand side of \eqref{3.20}, we write it as
\begin{equation}
\label{3.35}
\tilde g(x,y)=g(y)+\int_{-\infty}^y ds\,g(s)\,r(x;s,y)+\int_y^x ds\,g(s)\,r(x;s,y).
\end{equation}
We use the first and second lines of \eqref{3.3} in the integrals $\int_y^x$ and $\int_{-\infty}^y,$
respectively, in \eqref{3.35}, and we get
\begin{equation*}
\tilde g(x,y)=b_5+b_6+b_7,
\end{equation*}
where we have defined
\begin{equation}
\label{3.36}
b_5:=g(y)+\int_{-\infty}^y ds\,g(s)\,J\,\alpha(y,s)^\dagger J+\int_y^x ds\,g(s)\,\alpha(s,y),
\end{equation}
\begin{equation}
\label{3.37}
b_6:=\int_{-\infty}^y ds \int_y^x dt\,g(s)\,J\,\alpha(t,s)^\dagger J\,\alpha(t,y),
\end{equation}
\begin{equation}
\label{3.38}
b_7:=\int_y^x ds \int_s^x dt\,g(s)\,J\,\alpha(t,s)^\dagger J\,\alpha(t,y).
\end{equation}
The orders of the iterated integrals in \eqref{3.37} and \eqref{3.38} can be changed to $\int_y^x dt \int_{-\infty}^y ds$
and $\int_y^x dt\int_y^t ds,$ respectively. Using $\int_{-\infty}^y+\int_y^t = \int_{-\infty}^t,$ from \eqref{3.37}
and \eqref{3.38} we obtain
\begin{equation}
\label{3.39}
b_6+b_7=\int_y^x dt\int_{-\infty}^t ds\,g(s)\,J\,\alpha(t,s)^\dagger J\,\alpha(t,y).
\end{equation}
From \eqref{3.19} we observe that the sum of the first two terms on the right-hand side \eqref{3.36} is equal to $q(y).$
Interchanging the dummy integration variables $s$ and $t$ in \eqref{3.39}, we then get
\begin{equation*}
b_5+b_6+b_7=q(y)+\int_y^x ds\,g(s)\,\alpha(s,y)+ \int_y^x ds\int_{-\infty}^s dt\,g(t)\,J\,\alpha(s,t)^\dagger J\,\alpha(s,y),
\end{equation*}
or equivalently
\begin{equation}
\label{3.40}
\tilde g(x,y)=q(y)+\int_y^x ds \left[ g(s)+\int_{-\infty}^s dt\,g(t)\,J\,\alpha(s,t)^\dagger J \right]\alpha(s,y).
\end{equation}
From \eqref{3.19} we observe that the quantity inside the brackets in \eqref{3.40} is equal to  $q(x),$ and hence \eqref{3.23} holds.
\end{proof}

We remark that the transformed integral equation \eqref{3.17} looks the same as
the integral equation \eqref{2.21} in operator notation even though the domains of the
corresponding integral operators are different.

The Darboux transformation at the potential and wavefunction levels on $(-\infty,x)$ is obtained with the help of the following theorem.

\begin{theorem}
\label{theorem3.3}
Assume that \eqref{3.1} is uniquely solvable for $\alpha(x,y)$ in $\mathcal H^2$ and that the operator
$\Omega$ satisfies \eqref{2.6}. Also assume that $\tilde\Omega-\Omega$ corresponds to the finite-rank perturbation
given in \eqref{2.11}. Let $\tilde\alpha(x,y)$ be the solution to \eqref{3.16} and let $n(x),$ $\Gamma(x),$ and $\tilde g(x,y)$
be the quantities given in \eqref{3.18}, \eqref{3.21}, and \eqref{3.23}, respectively. Then, the Darboux
transformation at the wavefunction level is obtained from $\tilde\alpha(x,y)-\alpha(x,y),$ which can be
written in terms of $n(x)$, $\Gamma(x),$ and $\tilde g(x,y)$ as
\begin{equation}
\label{3.41}
\tilde\alpha(x,y)-\alpha(x,y)=-n(x)\,\Gamma(x)^{-1}\,\tilde g(x,y).
\end{equation}
Thus, $\tilde\alpha(x,y)$ is explicitly constructed from the unperturbed quantity $\alpha(x,y)$ and
the perturbation quantities $f(x)$ and $g(y)$ appearing in \eqref{2.11}. Furthermore, the Darboux transformation
at the potential level is obtained from $\tilde\alpha(x,x)-\alpha(x,x),$ which can be
written in terms of $n(x)$, $q(x),$ and $\Gamma(x)$ as
\begin{equation}
\label{3.42}
\tilde\alpha(x,x)-\alpha(x,x)=-n(x)\,\Gamma(x)^{-1}\,q(x).
\end{equation}
\end{theorem}

\begin{proof}
We obtain \eqref{3.41} directly from \eqref{3.33}, but by using the alternate expression
for $\tilde g(x,y)$ given in \eqref{3.23}. Note that from \eqref{3.23} it follows that
\begin{equation}
\label{3.43}
\tilde g(x,x)=q(x).
\end{equation}
Letting $y\to x^-$ in \eqref{3.41} and using \eqref{3.43}, we obtain \eqref{3.42}.
\end{proof}
 
We note the similarity between the pair of
equalities \eqref{2.16} and \eqref{2.17} and the pair \eqref{3.41} and \eqref{3.42}.
The four auxiliary quantities defined in \eqref{2.12}--\eqref{2.15} 
also have the same appearance as the corresponding four quantities
defined in \eqref{3.18}--\eqref{3.21}. The difference is that
the former quantities are related 
to the interval $(x,+\infty)$ and the latter quantities are related to $(-\infty,x).$
This is a crucial aspect of our generalized
method.

Having presented our generalized approach on $(-\infty,x)$ to the Darboux transformation, let us briefly
illustrate the difference between the standard method and our generalized method on a specific linear system, namely for the scalar
Schr\"odinger equation on the full line given in \eqref{2.22}. Let us add a bound state at $k=\text{i}\kappa$
with the dependency constant $\gamma$ to \eqref{2.22}. In this case, the relevant wavefunction is
the right Jost solution $f_r(k,x)$ to \eqref{2.22} with the asymptotics in \eqref{2.23}.
In the standard approach, the corresponding Darboux transformation formulas at the potential and wavefunction levels
are given by \eqref{2.25} and \eqref{2.27}, respectively. Our generalized approach in this case yields the following.
On the interval $(-\infty,x),$ the second equality in \eqref{2.11} is given by
\begin{equation}
\label{3.44}
\tilde\omega(x,y)-\omega(x,y)=c_r^2\,e^{\kappa(x+y)},
\end{equation}
where the norming constant
$c_r$ is related to the perturbed right Jost solution $\tilde f_r(k,x)$ as
\begin{equation*}
c_r:=\left[ \ds\int_{-\infty}^\infty dx\,\tilde f_r(i\kappa,x)^2   \right]^{-1/2}.
\end{equation*}
We remark that the dependency constant $\gamma$ is related
to $c_r$ \cite{AK2001} as
\begin{equation*}
\gamma=\ds\frac{c_r^2}{2\kappa\,T(i\kappa)},
\end{equation*}
where we recall that $T(k)$ is the transmission coefficient corresponding to the unperturbed potential $u(x)$ 
appearing in \eqref{2.22}. On the interval $(-\infty,x),$ comparing \eqref{2.11} and \eqref{3.44} we see that
the perturbation quantities $f(x)$ and $g(y)$ can be chosen as
\begin{equation*}
f(x)=c_r\,e^{\kappa x},\quad g(y)=c_r\,e^{\kappa y}.
\end{equation*}
The corresponding intermediate quantities $n(x),$ $q(y),$ $\Gamma(x),$ and $\tilde g(x,y)$ appearing in
\eqref{3.18}, \eqref{3.19}, \eqref{3.21}, and \eqref{3.20}, respectively, are evaluated as
\begin{equation}
\label{3.45}
n(x)=c_r\,f_r(i\kappa,x),\quad q(y)=c_r\,f_r(i\kappa,y),\quad 
\Gamma(x)=1+c_r^2 \ds\int_{-\infty }^x dz\,f_r(i\kappa,z)^2,
\end{equation}
\begin{equation*}
\tilde g(x,y)=c_r\,f_r(i\kappa,y)+c_r \ds\int_{-\infty}^x  dz\,f_r(i\kappa,z)\,\alpha(z,y),
\end{equation*}
and in this case we have
\begin{equation*}
\alpha(x,y)=\ds\frac{1}{2\pi}\ds\int_{-\infty }^\infty dk\left[f_r(k,x)-e^{-ikx}
\right]e^{iky},
\quad
\tilde\alpha(x,y)=\ds\frac{1}{2\pi}\ds\int_{-\infty}^\infty dk\left[\tilde f_r(k,x)-e^{-ikx}
\right]e^{iky},
\end{equation*}
\begin{equation}
\label{3.46}
f_r(k,x)=e^{-ikx}+  \ds\int_{-\infty}^x dy\,\alpha(x,y)\,e^{-iky},\quad 
\tilde f_r(k,x)=e^{-ikx}+  \ds\int_{-\infty}^x dy\,   \tilde\alpha(x,y)\,e^{-iky}.
\end{equation}
The counterpart of \eqref{2.25} is given by
\begin{equation}
\label{3.47}
\tilde u(x)-u(x)=-2\ds\,\frac{d^2\ln \left(\Gamma(x)\right)}{dx^2},
\end{equation}
where $\Gamma(x)$ is the quantity appearing in the last equality of \eqref{3.45},
and the counterpart of \eqref{2.27} is given by
\begin{equation}
\label{3.48}
\tilde f_r(k,x)-f_r(k,x)=-\ds\frac{c_r^2\,f_r(i\kappa,x)  \int_{-\infty}^x dy\, f_r(k,y)\,
f_r(i\kappa,y)}{1+c_r^2  \int_{-\infty}^x dz\, f_r(i\kappa,z)^2}.
\end{equation}
Analogous to \eqref{2.31}, 
we can write the integral appearing in the numerator in \eqref{3.48} as
\begin{equation*}
\int_{-\infty}^x dy\, f_r(k,y)\,
f_r(i\kappa,y)=\ds\frac{1}{k^2+\kappa^2}\left[f_r(k,x)\,f'_r(i\kappa,x)-f'_r(k,x)\,f_r(i\kappa,x)
\right],
\end{equation*}
and hence \eqref{3.48} can be expressed in the equivalent form as
\begin{equation*}
\tilde f_r(k,x)=\left[1-\ds\frac{c_r^2\,f_r(i\kappa,x) \,f'_r(i\kappa,x)}
{(k^2+\kappa^2)\left[1+c_r^2  \int_{-\infty}^x dz\, f_r(i\kappa,z)^2\right]}
\right] f_r(k,x)
+\ds\frac{c_r^2\,f_r(i\kappa,x)^2 \,f'_r(k,x)}
{(k^2+\kappa^2)\left[1+c_r^2  \int_{-\infty}^x dz\, f_r(i\kappa,z)^2\right]}.
\end{equation*}
We remark that \eqref{3.47} and \eqref{3.48} in our generalized method use
only the relevant wavefunction $f_r(k,x)$ and the bound-state
parameters $\kappa$ and $c_r.$ On the other hand,
the Darboux transformation formulas \eqref{2.25} and \eqref{2.27} in the standard method not only use
the relevant wavefunction $f_r(k,x)$ but also the information about the unperturbed left Jost solution $f_l(k,x)$ and the unperturbed transmission
coefficient $T(k).$
Furthermore, as already mentioned in Section~\ref{section2},
in the standard method there is the limitation that the bound state at $k=i\kappa$ must be added below the existing bound-state
energies in the discrete spectrum so that the quantity $\eta(x)$ defined in \eqref{2.24} is positive
for all $x.$ Our generalized method does not have such a limitation because $\Gamma(x)$
appearing in \eqref{3.45} is already positive for all $x.$

\section{The generalized method on the interval $(0,x)$ }
\label{section4}

Let us recall that we are
interested in obtaining the Darboux transformation formulas at the potential and wavefunction levels when the 
linear system $\mathcal L\Psi = \lambda \Psi$ is perturbed by changing only the discrete spectrum with
the addition or removal of a finite number of eigenvalues.
The primary goal in this section is to develop our generalized method on the interval $(0,x)$
for the Darboux transformation.
On the interval $(0,x)$ the relevant wavefunction usually satisfies some appropriate initial conditions at $x=0$ rather than
a spacial asymptotic condition at infinity. At the end of the section, we provide a comparison between our generalized method
and the standard method when they are applied on the half-line Schr\"odinger operator both in the Dirichlet and non-Dirichlet cases. 

Our generalized approach on $(0,x)$ to the Darboux transformation works as follows. 
We have the
unperturbed fundamental integral equation given by
\begin{equation}
\label{4.1}
\alpha(x,y) + \omega(x,y)+ \int_0^xdz \, \alpha(x,z) \, \omega(z,y)=0, \qquad 0<y<x,
\end{equation}
which corresponds to \eqref{2.2} on the interval $(0,x).$ The perturbed fundamental integral equation corresponding
to \eqref{2.10} on $(0,x)$ is given by
\begin{equation}
\label{4.2}
\tilde\alpha(x,y)+\tilde\omega(x,y)+\int_0^x dz \, \tilde\alpha(x,z)\, \tilde\omega(z,y)=0, \qquad 0<y<x.
\end{equation}
Our generalized method on $(0,x)$ involves the same steps outlined in Section~\ref{section2}.
The difference between $\tilde\omega(x,y)$ and $\omega(x,y)$ appearing in \eqref{4.2} and \eqref{4.1}, respectively, is
equal to $f(x)\,g(y),$ as indicated in \eqref{2.11}.
In terms of the unperturbed solution $\alpha(x,y)$ to \eqref{4.1} and the perturbation quantities $f(x)$ and $g(y),$ we construct
the key quantities $n(x),$ $q(y),$ $\tilde g(x,y),$ and $\Gamma(x),$ which are the analogs of the quantities appearing in \eqref{2.12},
\eqref{2.1}, \eqref{2.14}, and \eqref{2.15}, respectively. Next, we express the perturbed solution $\tilde\alpha(x,y)$ to \eqref{4.2}
in terms of those key quantities, as in \eqref{2.16}. Then, we get the Darboux transformation at the potential level with the
help of $\tilde\alpha(x,x)-\alpha(x,x)$ and obtain the Darboux transformation at the wavefunction level with the help of
$\tilde\alpha(x,y)-\alpha(x,y)$.

As done on the intervals $(x,+\infty)$ in Section~\ref{section2} and $(-\infty,x)$ in Section~\ref{section3}, we
assume that
the unperturbed fundamental integral equation \eqref{4.1} is uniquely solvable for $\alpha(x,y),$ and we express $\alpha(x,y)$ 
explicitly in terms of the corresponding resolvent kernel as
\begin{equation}
\label{4.3}
\alpha(x,y)=-\omega(x,y)-\int_0^x dz\,\omega(x,z)\,r(x;z,y), \qquad 0<y<x,
\end{equation}
which is the analog of \eqref{2.5} on $(x,+\infty)$ and \eqref{3.2} on $(-\infty,x).$

Analogous to Theorems~\ref{theorem2.1} and \ref{theorem3.1}, in the next theorem we relate the relevant resolvent
kernel $r(x;z,y)$ on the interval $(0,x)$ to the unique solution $\alpha(x,y)$ to \eqref{4.1},

\begin{theorem}
\label{theorem4.1}
Assume that \eqref{4.1} is uniquely solvable in the Hilbert space $\mathcal H^2.$ Also suppose that the operator
$\Omega$ with the kernel $\omega(x,y)$ appearing in \eqref{4.1} satisfies \eqref{2.6} and that $\Omega$ is related to
the resolvent $R$ with the kernel $r(x;z,y)$ as in \eqref{2.4}. We then have the following:
\begin{enumerate}
\item[\text{\rm(a)}] The resolvent $R$ and its kernel $r(x;z,y)$ satisfy \eqref{2.19}.

\item[\text{\rm(b)}] The resolvent kernel $r(x;z,y)$ is explicitly expressed in terms of the solution $\alpha(x,y)$ to \eqref{4.1} as
\begin{equation}
\label{4.4}
r(x;z,y)=\begin{cases}\alpha(z,y)+\ds\int_z^x ds\,J\,\alpha(s,z)^\dagger J\,\alpha(s,y), \qquad 0<y<z<x,\\
J\,\alpha(y,z)^\dagger J+\ds\int_y^x ds\,J\,\alpha(s,z)^\dagger J\,\alpha(s,y),\qquad 0<z<y<x,\end{cases}
\end{equation}
where we recall that $J$ is the involution matrix appearing in \eqref{2.6}.
\end{enumerate}
\end{theorem}
	
\begin{proof}
The proof is similar to the proof of Theorem~\ref{theorem3.1}. We refer the reader to
Proposition~4.1.1 and Theorem~4.1.2 of \cite{U2014} for (a) and (b), respectively, for the details of the proof.
\end{proof}

As Theorem~\ref{theorem4.1} indicates, the resolvent kernel $r(x;z,y)$ is uniquely determined by the solution
$\alpha(x,y)$ to \eqref{4.1}. Let us remark that, given the pair $\alpha(x,y)$ and $\omega(x,y),$ the integral equation \eqref{4.3}
may still be satisfied if the resolvent kernel $r(x;z,y)$ there is replaced by some other functions. However, such other
functions satisfying \eqref{4.3} cannot all be equal to the uniquely determined resolvent kernel and they do not satisfy
\eqref{4.4}. This important fact is illustrated later in Example~\ref{example5.9}.

Next, we introduce the key quantities $n(x)$ and $q(y)$ on the interval $(0,x),$ which are
the analogs of the quantities presented in \eqref{2.12} and \eqref{2.13}, respectively, on $(x,+\infty)$
and the analogs of those presented in \eqref{3.18} and \eqref{3.19}, respectively, on $(-\infty,x).$ We let
\begin{equation}
\label{4.5}
n(x):=f(x)+\int_0^x dz\, \alpha(x,z)\,f(z),
\end{equation}
\begin{equation}
\label{4.6}
q(y):=g(y)+\int_0^y dz\, g(z)\,J\,\alpha(y,z)^\dagger J,
\end{equation}
where we recall that $f(x)$ and $g(y)$ are the perturbation quantities appearing in \eqref{2.11}, the involution
matrix $J$ is as in \eqref{2.6}, and $\alpha(x,y)$ is the unique solution to \eqref{4.1}.

The following theorem is the analog of Theorem~\ref{theorem3.2} presented on the interval $(-\infty,x).$ It shows
that the integral equation \eqref{4.2} on $(0,x)$ can be transformed into another integral equation with a separable kernel.

\begin{theorem}
\label{theorem4.2}
Assume that \eqref{4.1} is uniquely solvable for $\alpha(x,y)$ in the Hilbert space $\mathcal H^2$ and that
the operator $\Omega$ with the kernel $\omega(x,y)$ satisfies \eqref{2.6}. Let $R$ be the corresponding resolvent
operator as in \eqref{2.4} with its kernel $r(x;z,y).$ Furthermore, assume that $\tilde\Omega-\Omega$
corresponds to the finite-rank perturbation given in \eqref{2.11} and let $f(x)$ and $g(y)$ be the quantities
appearing in \eqref{2.11}. Then, we have the following:
\begin{enumerate}
\item[\text{\rm(a)}] The linear integral equation \eqref{4.2} is transformed into the linear integral equation on the interval $(0,x)$ given by
\begin{equation}
\label{4.7}
\tilde\alpha(I+F\,\tilde G)=\alpha-f\tilde g,
\end{equation}
which is the analog of \eqref{3.17} presented on the interval $(-\infty,x).$ The kernel of the integral equation
\eqref{4.7} is given by $f(y)\,\tilde g(x,z),$ where $\tilde g(x,y)$ is defined as
\begin{equation}
\label{4.8}
\tilde g(x,y):=g(y)+\int_0^x dz\, g(z)\, r(x;z,y),
\end{equation}
which is the analog of \eqref{3.20} presented on the interval $(-\infty,x).$ The kernel $f(y)\,\tilde g(x,z)$ of the integral
equation \eqref{4.7} is separable in $y$ and $z.$ Consequently, \eqref{4.7} is explicitly solvable by the methods of linear
algebra, and the solution $\tilde\alpha(x,y)$ to \eqref{4.2} is
expressed as in \eqref{2.16}, where the relevant quantities $n(x),$ $\tilde g(x,y),$ and
$\Gamma(x)$ are now those defined on the interval $(0,x)$ rather than on $(x,+\infty).$ In particular, $n(x)$ now is
the quantity defined in \eqref{4.5}, $q(y)$ is the quantity defined in \eqref{4.6}, $\tilde g(x,y)$ is the quantity
defined in \eqref{4.8}, and the quantity $\Gamma(x)$ is the analog of the quantity in \eqref{2.15} and is now given by
\begin{equation}
\label{4.9}
\Gamma(x):=I+\int_0^x dz\,\tilde g(x,z)\,f(z).
\end{equation}
\item[\text{\rm(b)}] The solution $\tilde\alpha(x,y)$ to \eqref{4.2} expressed as in \eqref{2.16} can be written also as
\begin{equation}
\label{4.10}
\tilde\alpha(x,y)=\alpha(x,y)-n(x)\left[ I+\int_0^x ds\,\tilde g(x,s)\,f(s) \right]^{-1} \tilde g(x,y), \qquad 0<y<x.
\end{equation}
\item[\text{\rm(c)}] The quantity $\tilde g(x,y)$ defined in \eqref{4.8} can be expressed in terms of the
quantities $\alpha(x,y)$ and $q(y)$ appearing in \eqref{4.1} and \eqref{4.6}, respectively, and we have
\begin{equation}
\label{4.11}
\tilde g(x,y)=q(y)+\int_y^x dz\, q(z)\, \alpha(z,y),
\end{equation}
which is the analog of \eqref{2.14} but expressed on the interval $(0,x).$
\end{enumerate}
\end{theorem}

\begin{proof}
The proof is similar to the proof of Theorem~\ref{theorem3.2}. We refer the reader to Theorem~4.2.1, Theorem~4.2.2,
and Proposition~4.2.2 of \cite{U2014} for the details.
\end{proof}

The next theorem describes how to obtain the Darboux transformation at the potential and wavefunction levels 
in our generalized method on $(0,x).$

\begin{theorem}
\label{theorem4.3}
Assume that \eqref{4.1} is uniquely solvable for $\alpha(x,y)$ in the Hilbert space $\mathcal H^2$ and that the operator
$\Omega$ satisfies \eqref{2.6}. Also assume that $\tilde\Omega-\Omega$ corresponds to the finite-rank perturbation
given in \eqref{2.11}. Let $\tilde\alpha(x,y)$ be the solution to \eqref{4.2} and let $n(x),$ $q(y),$ $\Gamma(x),$ and $\tilde g(x,y)$
be the quantities given in \eqref{4.5}, \eqref{4.6}, \eqref{4.9}, and \eqref{4.11}, respectively. Then, the Darboux transformation at the
wavefunction level is obtained from
the difference $\tilde\alpha(x,y)-\alpha(x,y).$ We can explicitly express that difference
in terms of $n(x)$, $\Gamma(x),$ and $\tilde g(x,y)$ as
\begin{equation}
\label{4.12}
\tilde\alpha(x,y)-\alpha(x,y)=-n(x)\,\Gamma(x)^{-1}\,\tilde g(x,y), \qquad 0<y<x,
\end{equation}
which is the analog of \eqref{2.16} presented on the interval $(x,+\infty)$ and the analog of \eqref{3.41} presented
on $(-\infty,x).$ Therefore, $\tilde\alpha(x,y)$ on the interval $(0,x)$ is explicitly constructed from the unperturbed quantity
$\alpha(x,y)$ and the perturbation quantities $f(x)$ and $g(y)$ appearing in \eqref{2.11}. Furthermore, the Darboux transformation
at the potential level is obtained from $\tilde\alpha(x,x)-\alpha(x,x),$ which can be written explicitly in terms of
$n(x),$ $q(x),$ and $\Gamma(x)$ as
\begin{equation}
\label{4.13}
\tilde\alpha(x,x)-\alpha(x,x)=-n(x)\,\Gamma(x)^{-1}\,q(x).
\end{equation}

\end{theorem}

\begin{proof}
We obtain \eqref{4.12} directly from \eqref{4.9} and \eqref{4.10}. From \eqref{4.11} we see that
$\tilde g(x,x)=q(x),$ and hence the evaluation of \eqref{4.12} at $y=x^-$ yields \eqref{4.12}. 
\end{proof}

Having presented the details of our generalized approach on the interval $(0,x)$
for the Darboux transformation, we now briefly illustrate the difference between the standard method and our generalized method in 
two specific cases, one of which is the Schr\"odinger operator on the half line with the Dirichlet boundary condition and
the other is the corresponding operator with a non-Dirichlet boundary condition.

Let us first discuss the Dirichlet case. 
We assume that a bound state is added to the discrete spectrum at
$k=i\kappa$ with the norming constant $C$ for the scalar Schr\"odinger equation on the half line, which is given by
\begin{equation}
\label{4.14}
-\ds\frac{d^2\psi(k,x)}{dx^2}+u(x)\,\psi(k,x)=k^2\,\psi(k,x), \qquad x>0,
\end{equation}
with the Dirichlet boundary condition
\begin{equation}
\label{4.15}
\psi(0)=0.
\end{equation}
In this case, the relevant wavefunction is the regular solution $\varphi(k,x)$ to \eqref{4.14} satisfying the initial conditions
\begin{equation}
\label{4.16}
\varphi(k,0)=0, \quad \varphi'(k,0)=1.
\end{equation}
Let $\tilde u(x)$ and $\tilde\varphi(k,x)$ be the corresponding perturbed potential and the perturbed regular solution
for the Schr\"odinger equation
\begin{equation}
\label{4.17}
-\ds\frac{d^2\tilde\psi(k,x)}{dx^2}+\tilde u(x)\,\tilde\psi(k,x)=k^2\,\tilde\psi(k,x), \qquad x>0,
\end{equation}
with the Dirichlet boundary condition specified in \eqref{4.15}.
The perturbed regular solution $\tilde\varphi(k,x)$ to \eqref{4.17} satisfies the initial conditions
\begin{equation*}
\tilde\varphi(k,0)=0, \quad \tilde\varphi'(k,0)=1.
\end{equation*}
The norming constant $C$ is related to $\tilde\varphi(k,x)$ through the normalization
\begin{equation}
\label{4.18}
C=\left[\int_0^\infty dx\,\tilde\varphi(i\kappa,x)^2\right]^{-1/2}.
\end{equation}
In this case, the standard approach \cite{ASU2015,AW2006,CS1989,GL1955,L1986,M1987,NJ1955}
for the Darboux transformation is based on using the formulas
\begin{equation}
\label{4.19}
\tilde u(x)=u(x)+2\,\ds\frac{dA(x,x)}{dx},
\end{equation}
\begin{equation}
\label{4.20}
\tilde\varphi(k,x)=\varphi(k,x)+\int_0^x dy \, A(x,y) \, \varphi(k,y),
\end{equation}
where $A(x,y)$ satisfies the Gel'fand--Levitan integral equation
\begin{equation}
\label{4.21}
A(x,y) + M(x,y)+ \int_0^x dz \, A(x,z) \, M(z,y)=0, \qquad 0<y<x,
\end{equation}
with $M(x,y)$ defined as
\begin{equation}
\label{4.22}
M(x,y):=C^2\,\varphi(i\kappa,x)\,\varphi(i\kappa,y).
\end{equation}
Since $M(x,y)$ is separable in $x$ and $y,$ the Gel'fand--Levitan integral equation \eqref{4.21} can be
solved explicitly by the methods of linear algebra, and using the resulting solution $A(x,y)$
in \eqref{4.19} and \eqref{4.20} we obtain the Darboux transformation formulas at the potential and wavefunction levels,
respectively, as
\begin{equation}
\label{4.23}
\tilde u(x)-u(x)=-2\ds\frac{d}{dx}\left[\ds\frac{C^2\,\varphi(i\kappa,x)^2}{1+C^2  \int_0^x dz\,
\varphi(i\kappa,z)^2}\right],
\end{equation}
\begin{equation}
\label{4.24}
\tilde\varphi(k,x)-\varphi(k,x)=-\ds\frac{C^2\,\varphi(i\kappa,x)  \int_0^x dy\, \varphi(k,y)\,
\varphi(i\kappa,y)}{1+C^2  \int_0^x dz\, \varphi(i\kappa,z)^2}.
\end{equation}

In this case, our generalized method yields the same formulas listed in \eqref{4.23} and \eqref{4.24}
even though the starting point of our generalized method is different and is as follows.
In this case, on the interval $(0,x)$ the perturbation given in the second equality of \eqref{2.11} corresponds to 
\begin{equation}
\label{4.25}
\tilde\omega(x,y)-\omega(x,y)
=\frac{C^2}{\kappa^2}\sinh(\kappa x)\sinh(\kappa y).
\end{equation}
Comparing \eqref{4.25} with the second equality of \eqref{2.11}, we see that the perturbation quantities $f(x)$ and $g(y)$ 
can be chosen as
\begin{equation*}
f(x)=\frac{C}{\kappa}\sinh(\kappa x),\quad g(y)=\frac{C}{\kappa}\sinh(\kappa y).
\end{equation*}
The corresponding intermediate quantities $n(x),$ $q(y),$ $\Gamma(x),$ and $\tilde g(x,y)$ appearing in
\eqref{4.5}, \eqref{4.6}, \eqref{4.9}, and \eqref{4.11}, respectively, are evaluated as
\begin{equation}
\label{4.26}
n(x)=C\,\varphi(i\kappa,x),\quad q(y)=C\,\varphi(i\kappa,y),\quad 
\Gamma(x)=1+C^2 \ds\int_0^x dz\,\varphi(i\kappa,z)^2,
\end{equation}
\begin{equation}
\label{4.27}
\tilde g(x,y)=C\,\varphi(i\kappa,y)+C\ds\int_0^x  dz\, \varphi(i\kappa,z)\,\alpha(z,y),
\end{equation}
and in this case we have
\begin{equation}
\label{4.28}
\alpha(x,y)=\ds\frac{2}{\pi}\ds\int_0^\infty dk\left[\varphi(k,x)-\ds\frac{\sin (kx)}{k}
\right]k\,\sin(ky),
\end{equation}
\begin{equation}
\label{4.29}
\tilde\alpha(x,y)=\ds\frac{2}{\pi}\ds\int_0^\infty dk\left[\tilde\varphi(k,x)-\ds\frac{\sin (kx)}{k}
\right]k\,\sin(ky),
\end{equation}
\begin{equation}
\label{4.30}
\varphi(k,x)=\ds\frac{\sin (kx)}{k}+\ds\int_0^x dy\, \alpha(x,y)\,
\ds\frac{\sin (ky)}{k},
\end{equation}
\begin{equation}
\label{4.31}
\tilde\varphi(k,x)=\ds\frac{\sin (kx)}{k}+\ds\int_0^x dy\, \tilde\alpha(x,y)\,
\ds\frac{\sin (ky)}{k}.
\end{equation}
Then, the counterpart of \eqref{4.23} is obtained with the help of \eqref{4.13}, and it is given by
\begin{equation}
\label{4.32}
\tilde u(x)-u(x)=-2\ds\,\frac{d^2\ln \left(\Gamma(x)\right)}{dx^2},
\end{equation}
where $\Gamma(x)$ is the quantity appearing in the last equality of \eqref{4.26}.
Comparing \eqref{4.23} and \eqref{4.32}, we see that they agree with each other.
The counterpart of \eqref{4.24} is obtained with the help of \eqref{4.12}, and it
is the same as \eqref{4.24} itself. We note that the bound state
at $k=i\kappa$ can be added anywhere in the discrete spectrum because the quantity
$\Gamma(x)$ given in \eqref{4.26} is already positive.
Then, both the standard method and our generalized method yield the same 
Darboux transformation formulas even though their starting points are different.
In our generalized method, we use a fundamental integral equation for
each of the unperturbed and perturbed problems rather than
a fundamental integral equation used only for the perturbed problem in the
standard method based on the Gel'fand--Levitan theory.

Let us discuss the non-Dirichlet case and make a comparison between the standard method and our 
generalized method. A non-Dirichlet boundary condition associated with \eqref{4.14} is specified as
\begin{equation}
\label{4.33}
\psi'(0)+(\cot\theta)\,\psi(0)=0, \qquad \theta\in(0,\pi),
\end{equation}
for some fixed value of $\theta.$ In this case, the relevant wavefunction is the regular solution to \eqref{4.14}
satisfying the initial conditions
\begin{equation}
\label{4.34}
\varphi(k,0)=1, \quad \varphi'(k,0)=-\cot\theta.
\end{equation}
Again, we add a bound state at $k=i\kappa$ with the norming constant $C$ to obtain  \eqref{4.17} from \eqref{4.14}.
Let $\tilde u(x)$ be the corresponding perturbed potential appearing in \eqref{4.17},
and let $\tilde\varphi(k,x)$ be the perturbed regular solution
satisfying the initial conditions
\begin{equation}
\label{4.35}
\tilde\varphi(k,0)=1, \quad \tilde\varphi'(k,0)=-\cot\tilde\theta,
\end{equation}
for some $\tilde\theta\in(0,\pi).$
The norming constant $C$ is related to $\tilde\varphi(k,x)$ as in \eqref{4.18}, even though the regular solutions
in the Dirichlet and non-Dirichlet cases are different. Let us also remark that
the value of the bound-state parameter $\kappa$ in the non-Dirichlet case is different from
that in the Dirichlet case. In the standard approach, the Gel'fand--Levitan method summarized in \eqref{4.19}--\eqref{4.22} 
again yields the Darboux transformation \eqref{4.23} and \eqref{4.24} at the potential and wavefunction levels, respectively.
In the non-Dirichlet case, our generalized method yields the same formulas as \eqref{4.23} and \eqref{4.24} 
even though the starting point of our generalized method is different than the starting point in the standard method.
In this case, in our generalized method on the interval $(0,x),$ the perturbation given in the second equality of \eqref{2.11} corresponds to 
\begin{equation*}
\tilde\omega(x,y)-\omega(x,y)
=C^2\cosh(\kappa x)\cosh(\kappa y),
\end{equation*}
and hence we can choose the perturbation quantities $f(x)$ and $g(y)$ as
\begin{equation*}
f(x)=C\cosh(\kappa x),\quad g(y)=C\cosh(\kappa y).
\end{equation*}
The corresponding intermediate quantities $n(x),$ $q(y),$ $\Gamma(x),$ and $\tilde g(x,y)$ appearing in
\eqref{4.5}, \eqref{4.6}, \eqref{4.9}, and \eqref{4.11}, respectively, have the same forms given in 
\eqref{4.26} and \eqref{4.27}. On the other hand, in the non-Dirichlet case, instead of 
\eqref{4.28}--\eqref{4.31} we have
\begin{equation*}
\alpha(x,y)=\ds\frac{2}{\pi}\ds\int_0^\infty dk\left[\varphi(k,x)-\cos (kx)
\right]\cos(ky),
\quad
\tilde\alpha(x,y)=\ds\frac{2}{\pi}\ds\int_0^\infty dk\left[\tilde\varphi(k,x)-\cos (kx)
\right]\cos(ky),
\end{equation*}
\begin{equation}
\label{4.36}
\varphi(k,x)=\cos(kx)+\ds\int_0^x dy\, \alpha(x,y)\,
\cos(ky),
\quad
\tilde\varphi(k,x)=\cos(kx)+\ds\int_0^x dy\, \tilde\alpha(x,y)\,
\cos(ky).
\end{equation}
As already mentioned, our generalized method in the non-Dirichlet case yields the same formulas as \eqref{4.23} and \eqref{4.24}
obtained by the standard method 
for the Darboux transformation.
Our generalized method uses a fundamental integral equation for
each of the unperturbed and perturbed problems whereas the
standard method uses only one 
fundamental integral equation,
i.e. the Gel'fand--Levitan equation.

\section{Examples}
\label{section5}

In this section we illustrate the theory presented in the earlier sections
with some explicit examples. Various aspects of our generalized method for the Darboux
transformation are demonstrated on the interval $(-\infty,x)$ in the first five examples and
on the interval $(0,x)$ in the remaining four examples.

In the first example, we illustrate the construction of the resolvent kernel $r(x;z,y)$
corresponding to a particular operator $\Omega$ on the interval $(-\infty,x).$

\begin{example}\label{example5.1}
\normalfont
Let us illustrate the determination on $(-\infty,x)$ of the resolvent kernel $r(x;z,y)$ for the
operator
$\Omega$ whose kernel $\omega(z,y)$ is given by
\begin{equation}
\label{5.1}
\omega(z,y)=c_1^{2}\,e^{\kappa_1(z+y)},
\end{equation}
where $c_1$ and $\kappa_1$ are some positive parameters. To construct $r(x;z,y)$ from \eqref{5.1},
we can use \eqref{3.3} stated in Theorem~\ref{theorem3.1}(b). However, since $r(x;z,y)$ is unique
and \eqref{3.3} is equivalent to the two equations given in \eqref{3.8} and \eqref{3.9}, we can equivalently
construct $r(x;z,y)$ from \eqref{3.8} and \eqref{3.9}, which are valid for $y<z<x$ and
$z<y<x,$ respectively. Because of the separability of $\omega(z,y)$ in $z$ and $y,$ we are able
to solve the integral equations in \eqref{3.8} and \eqref{3.9} explicitly with the methods
of linear algebra. Using \eqref{5.1} as input to \eqref{3.8}, we obtain
\begin{equation}
\label{5.2}
  r(x;z,y) + c^2_1\,e^{\kappa_1(z+y)}+\int_{-\infty}^x ds \, r(x;z,s) \, c^2_1\,e^{\kappa_1(s+y)}=0, \qquad y<z<x.
\end{equation}
The explicit solution to \eqref{5.2} in the interval $y<z<x$ is given by
\begin{equation}
\label{5.3}
  r(x;z,y)=-\ds{\frac{c^2_1\,e^{\kappa_1(z+y)}}{1+\left(c^2_1/(2\kappa_1)\right)e^{2\kappa_1 x}}}.
\end{equation}
In this particular example, as seen from \eqref{5.1} we have $\omega(z,y)=\omega(y,z).$ Thus,
the integral equation in \eqref{3.9} for $z<y<x$ is the same as the integral
equation in \eqref{3.8} for $y<z<x.$ 
Therefore, $r(x;z,y)$ presented in \eqref{5.3} is the resolvent
kernel for the operator corresponding to \eqref{5.1} for $y<z<x$ as well as for $z<y<x.$
\end{example}

In the next example, we illustrate the fact that for a given pair $\alpha(x,y)$ and $\omega(x,y),$
\eqref{3.2} may still be satisfied if some other functions are substituted for $r(x;z,y)$  besides
the unique resolvent kernel associated with $\alpha(x,y)$ and $\omega(x,y).$

\begin{example}\label{example5.2}
\normalfont
To illustrate that \eqref{3.2} remains satisfied if $r(x;z,y)$ is replaced by some other function
which is not the resolvent kernel, we use the particular operator $\Omega$ studied in
Example~\ref{example5.1}. From $\omega(z,y)$ in \eqref{5.1}, we get the
corresponding kernel
$\omega(x,y)$ as
\begin{equation}
\label{5.4}
\omega(x,y)=c_1^{2}\,e^{\kappa_1(x+y)},
\end{equation}
where $c_1$ and $\kappa_1$ are some positive parameters. In order to determine the corresponding
$\alpha(x,y),$ we use \eqref{5.4} in \eqref{3.1} and get
\begin{equation}
\label{5.5}
  \alpha(x,y) + c^2_1\,e^{\kappa_1(x+y)}+\int_{-\infty}^x ds \, \alpha(x,s) \, c^2_1\,e^{\kappa_1(s+y)}=0.
\end{equation}
As a result of the separability of the integral kernel in \eqref{5.5}, its solution is obtained explicitly as
\begin{equation}
\label{5.6}
  \alpha(x,y)=-\ds{\frac{c^2_1\,e^{\kappa_1(x+y)}}{1+\left(c^2_1/(2\kappa_1)\right)e^{2\kappa_1 x}}},\qquad y<x.
\end{equation}
Let us recall from Example~\ref{example5.1} that the unique resolvent kernel corresponding to the
pair $\omega(x,y)$ and $\alpha(x,y)$ appearing in \eqref{5.4} and \eqref{5.6}, respectively, is the
quantity given in \eqref{5.3} both when $y<z<x$ and $z<y<x.$ We now demonstrate that \eqref{3.2}
still holds if we use \eqref{5.4}, \eqref{5.6}, and another substitute for $r(x;z,y)$ such as
\begin{equation}
\label{5.7}
r(x;z,y)=\ds{\frac{c^2_1\,e^{\kappa_1(z+y)}}{\left[1+\left(c^2_1/(2\kappa_1)\right)e^{2\kappa_1 z}\right]^2}},
\end{equation}
which is different from the integral kernel $r(x;z,y)$ given in \eqref{5.3}. Using \eqref{5.4} and \eqref{5.7}
as input, we evaluate the right-hand side of \eqref{3.2} as
\begin{equation*}
-c_1^{2}\,e^{\kappa_1(x+y)}+c_1^{4}\,e^{\kappa_1(x+y)}\,\int_{-\infty}^x dz\,
\ds{\frac{e^{2\kappa_1z}}{\left[1+\left(c^2_1/(2\kappa_1)\right)e^{2\kappa_1 z}\right]^2}},
\end{equation*}
which can be explicitly expressed as
\begin{equation*}
-c_1^{2}\,e^{\kappa_1(x+y)}+c_1^{4}\,e^{\kappa_1(x+y)}
\left[\ds{\frac{1}{c_1^2}}-\ds\frac{1}{c_1^2}\,\ds{\frac{1}
{1+\left(c^2_1/(2\kappa_1)\right)e^{2\kappa_1 x}}}\right],
\end{equation*}
which in turn is equal to the right-hand side of \eqref{5.6}. Thus, we have demonstrated that \eqref{3.2}
is not only satisfied by the input triplet \eqref{5.3}, \eqref{5.4}, \eqref{5.6} but also by the input triplet \eqref{5.4}, \eqref{5.6}, \eqref{5.7}.
\end{example}

Let us recall that the correspondence between $\alpha(x,y)$ and $\omega(x,y)$ is unique and that $\alpha(x,y)$
is determined by $\omega(x,y)$ by solving the integral equation \eqref{3.1}. In Example~\ref{example5.1}
we have illustrated the construction of the resolvent kernel $r(x;z,y)$ from \eqref{3.8} and \eqref{3.9} using
$\omega(x,y)$ as input. In the next example, we illustrate the construction of $r(x;z,y)$ from \eqref{3.3} using $\alpha(x,y)$ as input.

\begin{example}
\label{example5.3}
\normalfont
In this example, we demonstrate the construction of the resolvent kernel $r(x;z,y)$ of \eqref{5.3},
by using $\alpha(x,y)$ of \eqref{5.6} as input to \eqref{3.3}. Since $\alpha(x,y)$ in \eqref{5.6}
is real valued and a scalar, we observe that $J\,\alpha(x,y)^\dagger J$ appearing in \eqref{3.3} is the same as
$\alpha(x,y)$ itself no matter whether we use $J=1$ or $J=-1.$ Thus, in this particular case,
\eqref{3.3} is equivalent to
\begin{equation}
\label{5.8}
r(x;z,y)=\begin{cases}\alpha(z,y)+\ds\int_z^x ds\,\alpha(s,z)\,\alpha(s,y), \qquad y<z<x,\\
\noalign{\medskip}
\alpha(y,z)+\ds\int_y^x ds\,\alpha(s,z)\,\alpha(s,y),\qquad z<y<x.\end{cases}
\end{equation}
Using \eqref{5.6} in the first line of \eqref{5.8}, for $y<z<x$ we obtain
\begin{equation}
\label{5.9}
r(x;z,y)=-\ds{\frac{c^2_1\,e^{\kappa_1(z+y)}}{1+\left(c^2_1/(2\kappa_1)\right)e^{2\kappa_1 z}}}
+\ds\int_z^x ds\,\ds{\frac{c^2_1\,e^{\kappa_1(s+z)}\,c^2_1\,e^{\kappa_1(s+y)}}{\left[1+\left(c^2_1/(2\kappa_1)
\right)e^{2\kappa_1 s}\right]^2}}.
\end{equation}
The integral in \eqref{5.9} can be explicitly evaluated, and after some simplifications we get
\begin{equation}
\label{5.10}
 r(x;z,y)=-\ds{\frac{c^2_1\,e^{\kappa_1(z+y)}}{1+\left(c^2_1/(2\kappa_1)\right)e^{2\kappa_1x}}},\qquad y<z<x.
\end{equation}
We observe that the right-hand sides of \eqref{5.3} and \eqref{5.10} coincide, and hence the resolvent kernel
constructed from $\alpha(x,y)$ for $y<z<x$ agrees with the resolvent kernel constructed from $\omega(x,y).$
To show that there is also the agreement when $z<y<x,$ we use \eqref{5.6} in the second line of \eqref{5.8} and obtain
\begin{equation}
\label{5.11}
r(x;z,y)=-\ds{\frac{c^2_1\,e^{\kappa_1(y+z)}}{1+\left(c^2_1/(2\kappa_1)\right)e^{2\kappa_1 y}}}
+\ds\int_y^x ds\,\ds{\frac{c^2_1\,e^{\kappa_1(s+z)}\,c^2_1\,e^{\kappa_1(s+y)}}{\left[1+\left(c^2_1/(2\kappa_1)
\right)e^{2\kappa_1s}\right]^2}}.
\end{equation}
The integral on the right-hand side of \eqref{5.11} can be explicitly evaluated, and from \eqref{5.11} we get
\begin{equation*}
 r(x;z,y)=-\ds{\frac{c^2_1\,e^{\kappa_1(z+y)}}{1+\left(c^2_1/(2\kappa_1)\right)e^{2\kappa_1x}}},\qquad z<y<x,
\end{equation*}
which also agrees with \eqref{5.3} when $z<y<x.$
Thus, we have demonstrated that the value of the resolvent kernel $r(x;z,y)$ constructed from $\alpha(x,y)$
by using \eqref{3.3} agrees with the value $r(x;z,y)$ constructed from $\omega(x,y)$ by using \eqref{3.8} and \eqref{3.9}.
\end{example}

Let us recall that we assume that the quantity $\omega(x,y)$ satisfies the $J$-symmetry stated in the
second equality in \eqref{2.6}. Consequently, the corresponding resolvent kernel satisfies the same $J$-symmetry
stated in the second equality in \eqref{2.19}, as indicated by Theorems~\ref{theorem2.1}, \ref{theorem3.1}, and \ref{theorem4.1}
in the intervals $(x,+\infty),$ $(-\infty,x),$ and $(0,x),$ respectively. However, the quantity $\alpha(x,y)$ corresponding
to $\omega(x,y)$ and $r(x;z,y)$ does not satisfy the $J$-symmetry. Thus, in general we have
\begin{equation}
\label{5.12}
 \alpha(x,y)\ne J\,\alpha(y,x)^\dagger J.
\end{equation}
When the support property of $\alpha(x,y)$ is taken into account, it is clear that we cannot have 
an equality in \eqref{5.12} because one side of that equality would be zero.
On the other hand, in practice one may first obtain the expression for $\alpha(x,y)$ without its support property
and may impose the support property afterwards. In such a case, the fact expressed in \eqref{5.12} is relevant,
and one must be careful in using
\eqref{2.20}, \eqref{3.3}, and \eqref{4.4} in the evaluation of the corresponding resolvent kernel
$r(x;z,y)$ on the intervals $(x,+\infty),$ $(-\infty,x),$ and $(0,x),$ respectively.

In the next example, we illustrate the fact that, when its support property is not taken into account, the solution $\alpha(x,y)$ to the fundamental integral equation
\eqref{3.1} does not satisfy the $J$-symmetry.

\begin{example}
\label{example5.4}
\normalfont
In this example we illustrate \eqref{5.12} in the scalar case and when $\alpha(x,y)$ is real valued.
In that case, \eqref{5.12} is equivalent to having
\begin{equation}
\label{5.13}
\alpha(x,y)\neq \alpha(y,x).
\end{equation}
In fact, the quantity $\alpha(x,y)$ given in \eqref{5.6} of Example~\ref{example5.2} readily conforms to \eqref{5.13}.
Let us elaborate on this point. In Example~\ref{example5.3}, contrary to \eqref{5.13}, if we had $\alpha(x,y)=
\alpha(y,x),$ from \eqref{3.3} we would get
\begin{equation}
\label{5.14}
r(x;z,y)=\begin{cases}\alpha(z,y)+\ds\int_z^x ds\,\alpha(z,s)\,\alpha(s,y), \qquad y<z<x,\\
\noalign{\medskip}
\alpha(z,y)+\ds\int_y^x ds\, \alpha(z,s)\,\alpha(s,y),\qquad z<y<x,\end{cases}
\end{equation}
instead of \eqref{5.8}. Using the value of $\alpha(x,y)$ given in \eqref{5.6} as input to \eqref{5.14}, we would obtain
for $y<z<x$ the expression
\begin{equation}
\label{5.15}
r(x;z,y)=-\ds{\frac{c^2_1\,e^{\kappa_1(z+y)}}{1+\left(c^2_1/(2\kappa_1)\right)e^{2\kappa_1z}}}
+\ds{\frac{c^4_1\,e^{\kappa_1(z+y)}}{1+\left(c^2_1/(2\kappa_1)\right)e^{2\kappa_1z}}}\int_z^x ds\,\ds{\frac{e^{2\kappa_1s}}
{1+\left(c^2_1/(2\kappa_1)\right)e^{2\kappa_1s}}},
\end{equation}
and for $z<y<x$ we would have
\begin{equation}
\label{5.16}
r(x;z,y)=
-\ds{\frac{c^2_1\,e^{\kappa_1(z+y)}}{1+\left(c^2_1/(2\kappa_1)\right)e^{2\kappa_1z}}}
+\ds{\frac{c^4_1\,e^{\kappa_1(z+y)}}{1+\left(c^2_1/(2\kappa_1)\right)e^{2\kappa_1 z}}}\int_y^x ds\,\ds{\frac{e^{2\kappa_1s}}
{1+\left(c^2_1/(2\kappa_1)\right)e^{2\kappa_1s}}}.
\end{equation}
The integrals on the right-hand sides of \eqref{5.15} and \eqref{5.16} can be explicitly evaluated, and hence \eqref{5.14}
would yield
\begin{equation}
\label{5.17}
r(x;z,y)=\begin{cases}
\ds{\frac{c^2_1\,e^{\kappa_1(z+y)}\big[-1+\ln(c^2_1\,e^{2\kappa_1x}+2\kappa_1)-\ln(c^2_1
\,e^{2\kappa_1z}+2\kappa_1)\big]}{1+\left(c^2_1/(2\kappa_1)\right)e^{2\kappa_1z}}}, \qquad y<z<x,\\
\noalign{\medskip}
\ds{\frac{c^2_1\,e^{\kappa_1(z+y)}\big[-1+\ln(c^2_1\,e^{2\kappa_1x}+2\kappa_1)-\ln(c^2_1
\,e^{2\kappa_1y}+2\kappa_1)\big]}{1+\left(c^2_1/(2\kappa_1)\right)e^{2\kappa_1z}}},\qquad z<y<x,
\end{cases}
\end{equation}
which contradicts the correct expression for $r(x;z,y)$ given in \eqref{5.3}. Since the resolvent kernel 
for any given $\alpha(x,y)$ must be unique, we conclude that the right-hand side of \eqref{5.17}
cannot be the correct expression for the resolvent kernel $r(x;z,y)$ corresponding
to 
$\alpha(x,y)$ in \eqref{5.6}.
\end{example}

In the next example, we use our generalized method on the interval $(-\infty,x)$ in order to illustrate the Darboux
transformation for the full-line Schr\"odinger equation \eqref{2.22}.

\begin{example}
\label{example5.5}
\normalfont
Corresponding to 
the Schr\"odinger equation \eqref{2.22}, let us assume that the unperturbed integral kernel
$\omega(x,y)$ appearing in \eqref{3.1} is equal to the quantity in \eqref{5.4}.
We already know from Example~\ref{example5.2} that the corresponding solution
$\alpha(x,y)$ to \eqref{3.1} is given by the
expression in \eqref{5.6}. Using \eqref{5.6} in the first equality of \eqref{3.46} we 
construct the corresponding wavefunction $f_r(k,x)$ as
\begin{equation}
\label{5.18}
f_r(k,x)=e^{-ikx}\left[1-\ds{\frac{ic^2_1\,e^{2\kappa_1x}}{\left(k+i\kappa_1\right)
\left[1+\left(c^2_1/(2\kappa_1)\right)e^{2\kappa_1x}\right]}}\right].
\end{equation}
Using \eqref{5.18} in \eqref{2.22}, we evaluate the unperturbed potential $u(x)$ as
\begin{equation}
\label{5.19}
u(x)=-\ds{\frac{4c^2_1\,\kappa_1\,e^{2\kappa_1x}}
{\big[1+\left(c^2_1/(2\kappa_1)\right)e^{2\kappa_1x}\big]^2}}.
\end{equation}
In this case, the unperturbed potential $u(x)$
has one bound state at $k=i\kappa_1$ with the norming constant $c_1.$
Let us now introduce the perturbation by adding an additional bound state at $k=i\kappa_2$ with the norming constant $c_2.$
The perturbation in \eqref{2.11} is then described by
\begin{equation}
\label{5.20}
\tilde\omega(x,y)-\omega(x,y)=c^2_2\,e^{\kappa_2(x+y)}.
\end{equation}
Comparing \eqref{2.11} and \eqref{5.20}, we see that $f(x)$ and $g(y)$ appearing in \eqref{2.11} can be chosen as
\begin{equation}
\label{5.21}
f(x)=c_2\,e^{\kappa_2x}, \quad  g(y)=c_2\,e^{\kappa_2y}.
\end{equation}
Using $\alpha(x,y)$ from \eqref{5.6} and $f(x)$ and $g(y)$ from \eqref{5.21}, we evaluate the quantities $n(x)$
and $q(y)$ defined in \eqref{3.18} and \eqref{3.19}, respectively, as
\begin{equation}
\label{5.22}
n(x)=c_2\,e^{\kappa_2x}-\ds\frac{c^2_1\,c_2\,e^{(2\kappa_1+\kappa_2)x}}{(\kappa_1+\kappa_2)
\big[1+\left(c^2_1/(2\kappa_1)\right) e^{2\kappa_1x}\big]},
\end{equation}
\begin{equation}
\label{5.23}
q(y)=c_2\,e^{\kappa_2y}-\ds\frac{c^2_1\,c_2\,e^{(2\kappa_1+\kappa_2)y}}{(\kappa_1+\kappa_2)
\big[1+\left(c^2_1/(2\kappa_1)\right)e^{2\kappa_1y}\big]}.
\end{equation}
Using \eqref{5.6} and \eqref{5.23}, we construct the quantity $\tilde g(x,y)$ given in \eqref{3.23} as
\begin{equation}
\label{5.24}
\tilde g(x,y)=c_2\,e^{\kappa_2 y} -\ds\frac{c_1^2\,c_2\,e^{\kappa_1(x+y)+\kappa_2 x }}
{(\kappa_1+\kappa_2)\big[1+\left(c^2_1/(2\kappa_1)\right) e^{2\kappa_1x}\big]}.
\end{equation}
Next, using \eqref{5.21} and \eqref{5.24} we evaluate the quantity $\Gamma(x)$ defined in \eqref{3.21} as
\begin{equation}
\label{5.25}
\Gamma(x)=1+\ds{\frac{c^2_2\,e^{2\kappa_2x}\big[c^2_1(\kappa_1-\kappa_2)^2\,
e^{2\kappa_1x}+2\kappa_1(\kappa_1+\kappa_2)^2\big]}{2\kappa_2(\kappa_1+\kappa_2)^2\big[c^2_1\,e^{2\kappa_1x}
+2\kappa_1\big]}}.
\end{equation}
Then, using \eqref{5.6}, \eqref{5.22}, \eqref{5.24}, and \eqref{5.25} in \eqref{3.22}, we get $\tilde\alpha(x,y)$ as
\begin{equation}
\label{5.26}
\tilde\alpha(x,y)=-\frac{Q_1+Q_2}{Q_3},
\end{equation}
where we have defined
\begin{equation*}
Q_1:=4(\kappa_1+\kappa_2)^2\kappa_1\,\kappa_2\left[c^2_1\,e^{\kappa_1(x+y)}+
c^2_2\,e^{\kappa_2(x+y)}\right],
\end{equation*}
\begin{equation*}
Q_2:=2(\kappa^2_1-\kappa^2_2)\,c^2_1\,c^2_2\left[\kappa_1e^{2\kappa_2x+\kappa_1(x+y)}-\kappa_2
e^{2\kappa_1x+\kappa_2(x+y)}\right],
\end{equation*}
\begin{equation*}
Q_3:=2\kappa_1(\kappa_1+\kappa_2)^2\left(c^2_2\,e^{2\kappa_2x}+2\kappa_2\right)
+c^2_1\,e^{2\kappa_1x}\left[c^2_2\,e^{2\kappa_2x}
(\kappa_1-\kappa_2)^2+2\kappa_2(\kappa_1+\kappa_2)^2\right].
\end{equation*}
Having $\tilde\omega(x,y)$ given in \eqref{5.26} at hand, we can construct the perturbed wavefunction $\tilde f_r(k,x)$
by using \eqref{5.26} in the second equality of \eqref{3.46}.
We obtain
\begin{equation*}
\tilde f_r(k,x)=e^{-i\,k\,x}\left[1+\frac{Q_4+Q_5}{(k+i\kappa_1)(k+i\kappa_2)\,Q_6}\right],
\end{equation*}
where we have let
\begin{equation*}
Q_4:=-2ik\,c^2_1\,c^2_2\,(\kappa_1+\kappa_2)(\kappa_1-\kappa_2)^2\,e^{2(\kappa_1+\kappa_2)x},
\end{equation*}
\begin{equation*}
Q_5:=4\kappa_1\,\kappa_2(\kappa_1+\kappa_2)^2\left[(\kappa_1-ik)c^2_2\,e^{2\kappa_2x}+(\kappa_2-ik)c^2_1\,e^{2\kappa_1x}\right],
\end{equation*}\begin{equation*}
Q_6:=c^2_1\,c^2_2(\kappa_1-\kappa_2)^2e^{2(\kappa_1+\kappa_2)x}
            +2(\kappa_1+\kappa_2)^2\left[\kappa_1\,c^2_2\,e^{2\kappa_2x}+\kappa_2\,c^2_1
\,e^{2\kappa_1x}+2\kappa_1\,\kappa_2\right].
\end{equation*}
In this case, the potential perturbation $\tilde u(x)-u(x)$ is related to
the quantity $\tilde\alpha(x,x)-\alpha(x,x)$ as
\begin{equation}
\label{5.27}
\tilde u(x)-u(x)=2\frac{d}{dx}\left[\tilde\alpha(x,x)-\alpha(x,x)\right].
\end{equation}
Hence, using \eqref{5.6}, \eqref{5.19}, and \eqref{5.26} in 
\eqref{5.27}, we construct the perturbed potential $\tilde u(x)$ as
\begin{equation*}
\tilde u(x)=\frac{16(\kappa_1+\kappa_2)^2\,(Q_7+Q_8)}{Q_9^2},
\end{equation*}
where we have defined
\begin{equation*}
Q_7:=
   4 c_2^2\,\kappa_1^2\,  \kappa_2^3 \,(\kappa_1 + \kappa_2)^2\,e^{2 \kappa_2 x} +
c_1^4 \,c_2^2\, \kappa_2^3 \,(\kappa_1-\kappa_2)^2\,e^{(4 \kappa_1 +2 \kappa_2) x},
\end{equation*}
\begin{equation*}
Q_8:=
4 c_1^2 \,\kappa_1^3\, \kappa_2^2\,(\kappa_1+\kappa_2)^2\,e^{2\kappa_1 x} +
4 c_1^2 \,c_2^2\,\kappa_1\, \kappa_2\,(\kappa_1^2-\kappa_2^2)^2\,e^{2(\kappa_1+\kappa_2) x}+
c_1^2 \,c_2^4\,\kappa_1^3\, (\kappa_1-\kappa_2)^2\,e^{(2\kappa_1 +4 \kappa_2) x},
\end{equation*}
\begin{equation*}
Q_9:=2\kappa_1(\kappa_1+\kappa_2)^2\left(c^2_2\,e^{2\kappa_2x}+2\kappa_2\right)+c^2_1\,e^{2\kappa_1x}
\left[c^2_2\,e^{2\kappa_2x}(\kappa_1-\kappa_2)^2+2\kappa_2(\kappa_1+\kappa_2)^2\right].
\end{equation*}

\end{example}

Having illustrated our generalized method for the Darboux transformation on the interval $(-\infty,x),$
we now turn to the demonstration of our method on the interval $(0,x).$ In the next example,
we illustrate the construction of the solution $\alpha(x,y)$ to the unperturbed fundamental integral equation \eqref{4.1}.

\begin{example}
\label{example5.6}
\normalfont
Let us illustrate the recovery of the solution $\alpha(x,y)$ to the integral equation \eqref{4.1}
corresponding to the operator $\Omega$ with the kernel specified as
\begin{equation}
\label{5.28}
\omega(x,y)=\frac{C^2_1}{\kappa^2_1}\,\sinh(\kappa_1x)\,\sinh(\kappa_1y).
\end{equation}
Since the kernel given in \eqref{5.28} is separable in $x$ and $y,$ we can explicitly solve the integral
equation \eqref{4.1} with input \eqref{5.28}, and we evaluate its solution as
\begin{equation}
\label{5.29}
\alpha(x,y)=-\ds\frac{\left( C^2_1/\kappa^2_1\right) \sinh(\kappa_1x)\,\sinh(\kappa_1y)}{1+\left(C_1^2/\kappa^2_1\right)
\big[\left(1/(4\kappa_1) \right)\sinh(2\kappa_1x)-x/2\big]
}.
\end{equation}

\end{example}

In the next example, we illustrate Theorem~\ref{theorem4.1}(b), namely the construction of the
resolvent kernel $r(x;z,y)$ appearing in \eqref{4.4} by using the solution $\alpha(x,y)$ to the integral equation \eqref{4.1}.

\begin{example}
\label{example5.7}
\normalfont
On the interval $(0,x),$ in order to illustrate the construction of the resolvent kernel $r(x;z,y)$ from \eqref{4.4},
in this example
we use the quantity $\alpha(x,y)$ in \eqref{5.29} as input. Since $\alpha(x,y)$ in \eqref{5.29}
is real valued and a scalar, we observe that $J\,\alpha(x,y)^\dagger J$ appearing in \eqref{4.4} is equal to 
$\alpha(x,y)$ itself whether we use $J=1$ or $J=-1.$ 
Thus, in this case, \eqref{4.4} is equivalent to
\begin{equation}
\label{5.30}
r(x;z,y)=\begin{cases}\alpha(z,y)+\ds\int_z^x ds\,\alpha(s,z)\,\alpha(s,y), \qquad 0<y<z<x,\\
\noalign{\medskip}
\,\alpha(y,z)+\ds\int_y^x ds\,\alpha(s,z)\,\alpha(s,y),\qquad 0<z<y<x.\end{cases}
\end{equation}
Using \eqref{5.29} in the first line of \eqref{5.30}, for $0<y<z<x$ we obtain
\begin{equation}
\label{5.31}
\begin{split}
r(x;z,y)=&
-\ds\frac{\left(C_1^2/\kappa^2_1\right) \sinh(\kappa_1 z)\,\sinh(\kappa_1 y)}{1+\left(C_1^2/\kappa^2_1\right)
\big[\left(1/(  4\kappa_1) \right)\sinh(2\kappa_1z)-z/2\big]}
\\
&+\ds\int_z^x ds\,\ds\frac{\left(C^4_1/\kappa^4_1\right) \,\sinh^2(\kappa_1s)\,\sinh(\kappa_1z)\sinh(\kappa_1y)}
{\big[1+\left(C_1^2/\kappa^2_1\right)
\left[\left(1/(  4\kappa_1) \right)\sinh(2\kappa_1 s)-s/2\right]\big]^2}.
\end{split}
\end{equation}
The integral on the right-hand side of \eqref{5.31} can be explicitly evaluated, and from \eqref{5.31} we get
\begin{equation}
\label{5.32}
 r(x;z,y)=-\ds{\frac{4C^2_1\,\kappa_1\,\sinh(\kappa_1z)\sinh(\kappa_1y)}{4\kappa^3_1-2C^2_1\,\kappa_1x+C^2_1\sinh(2\kappa_1x)}},
\end{equation}
for $0<y<z<x.$ 
Similarly, by using \eqref{5.29} in the second line of \eqref{5.30}, for $0<z<y<x$ we obtain $r(x;z,y).$
We determine that the expression for $r(x;z,y)$ in \eqref{5.32} holds both for $0<y<z<x$ and $0<z<y<x.$
\end{example}

In the next example, we illustrate our generalized method on the interval $(0,x)$ for the Darboux
transformation on the half-line Schr\"odinger equation with the Dirichlet boundary condition.

\begin{example}
\label{example5.8}
\normalfont
Consider the half-line Schr\"odinger equation in \eqref{4.14} with the Dirichlet boundary condition \eqref{4.15}. 
Let us assume that the unperturbed kernel $\omega(x,y)$ is specified as the quantity appearing in \eqref{5.28}.
In \eqref{5.29} of Example~\ref{example5.6} we have obtained the corresponding $\alpha(x,y),$ and in \eqref{5.32}
of Example~\ref{example5.7} we have determined the corresponding resolvent kernel $r(x;z,y).$
Using \eqref{5.29} as input to \eqref{4.30}
we construct the corresponding unperturbed wavefunction
$\varphi(k,x)$ as
\begin{equation}
\label{5.33}
 \varphi(k,x)=\ds\frac{\sin(kx)}{k}-\ds{\frac{4C^2_1\,\kappa_1\,\sinh^2(\kappa_1x)\big[\kappa_1\,\coth(\kappa_1x)\,\sin(kx)-k
\,\cos(kx)\big]}{k(k^2+\kappa^2_1)\big[4\,\kappa^3_1-2C^2_1\,\kappa_1\,x+C^2_1\,
\sinh(2\kappa_1x)\big]}}.
\end{equation}
As indicated in Section~\ref{section4}, $\varphi(k,x)$ is the regular solution to \eqref{4.14} satisfying \eqref{4.16}.
Using \eqref{5.33} in \eqref{4.14}, we evaluate the unperturbed potential $u(x)$ as
\begin{equation}
\label{5.34}
u(x)=\ds\frac{32C^2_1\,\kappa^2_1\,\sinh(\kappa_1x)\big[\left(-2\kappa_1^3+C_1^2\,\kappa_1x\right)\cosh(\kappa_1x)-C_1^2\,
\sinh(\kappa_1x)\big]}{\big[4\kappa_1^3-2C_1^2\,\kappa_1x+C_1^2\,\sinh(2\kappa_1x)\big]^2}.
\end{equation}
Let us add a bound state at $k=i\kappa_2$ with the norming constant $C_2$
to the unperturbed problem. This corresponds to the perturbation
\eqref{2.11}, which is given by
\begin{equation}
\label{5.35}
\tilde\omega(x,y)-\omega(x,y)=
\frac{C^2_2}{\kappa^2_2}\,\sinh(\kappa_2x)\,
\sinh(\kappa_2y),
\end{equation}
where we recall that $\omega(x,y)$ is the quantity in \eqref{5.28}.
Comparing \eqref{2.11} and \eqref{5.35}, we see that the quantities $f(x)$ and $g(y)$ in \eqref{2.11} can be chosen as
\begin{equation}
\label{5.36}
f(x)=\frac{C_2}{\kappa_2}\,\sinh(\kappa_2x), \quad  g(y)=\frac{C_2}{\kappa_2}\,\sinh(\kappa_2y).
\end{equation}
Using $\alpha(x,y)$ from \eqref{5.29} and $f(x)$ and $g(y)$ from \eqref{5.36}, we construct the quantities $n(x)$ and
$q(y)$ defined in \eqref{4.5} and \eqref{4.6}, respectively, as
\begin{equation}
\label{5.37}
n(x)=\frac{C_2\sinh(\kappa_2x)}{\kappa_2}-\frac{4C_1^2\,C_2\,\kappa_1\,\sinh(\kappa_1x)\big[\kappa_1\cosh(\kappa_1x)\sinh(\kappa_2x)
-\kappa_2\sinh(\kappa_1x)\cosh(\kappa_2x)\big]}
{(\kappa_1^2-\kappa_2^2)\big[4\kappa_1^3\,\kappa_2-2C_1^2\,\kappa_1\,\kappa_2\,x+C_1^2\,\kappa_2\,\sinh(2\kappa_1x)\big]},
\end{equation}
\begin{equation}
\label{5.38}
q(y)=\frac{C_2\sinh(\kappa_2y)}{\kappa_2}-\frac{4C_1^2\,C_2\,\kappa_1\,\sinh(\kappa_1y)\big[\kappa_1\cosh(\kappa_1y)\sinh(\kappa_2y)
-\kappa_2\sinh(\kappa_1y)\cosh(\kappa_2y)\big]}
{(\kappa_1^2-\kappa_2^2)\big[4\kappa_1^3\,\kappa_2-2C_1^2\,\kappa_1\,\kappa_2\,y+C_1^2\,\kappa_2\,\sinh(2\kappa_1y)\big]}.
\end{equation}
Similarly, using \eqref{5.29} and \eqref{5.38}, we obtain the quantity $\tilde g(x,y)$ of \eqref{4.11} as
\begin{equation}
\label{5.39}
\tilde g(x,y)=\frac{Q_{10}+Q_{11}}{Q_{12}},
\end{equation}
where we have let
\begin{equation*}
Q_{10}:=4C^2_1\,C_2\,\kappa_1\sinh(\kappa_1y) \left[\kappa_2\,\sinh(\kappa_1x)\cosh(\kappa_2x)
-\kappa_1\,\cosh(\kappa_1x)\sinh(\kappa_2x)\right],
\end{equation*}
\begin{equation*}
Q_{11}:=C_2\left(\kappa^2_1-\kappa^2_2\right)\sinh(\kappa_2y)
\left[4\kappa^3_1-2\,C^2_1\,\kappa_1x+C^2_1\sinh(2\kappa_1x)\right],
\end{equation*}
\begin{equation*}
Q_{12}:=\kappa_2\left(\kappa_1^2-\kappa^2_2\right)\left[4\kappa^3_1-2\,C^2_1\,\kappa_1x+C^2_1\sinh(2\kappa_1x)\right].
\end{equation*}
Next, using \eqref{5.36} and \eqref{5.39} we construct the quantity $\Gamma(x)$ of \eqref{4.9} as
\begin{equation}
\label{5.40}
\Gamma(x)=1-\frac{Q_{13}+Q_{14}+Q_{15}}{Q_{16}},
\end{equation}
where we have defined
\begin{equation*}
Q_{13}:=16C^2_1\,C^2_2\,\kappa_1\,\kappa_2\left[
\kappa_2^2 \sinh^2(\kappa_1x)\cosh^2(\kappa_2x)+\kappa_1^2
\,\cosh^2(\kappa_1x)\sinh^2(\kappa_2x)\right],
\end{equation*}
\begin{equation*}
Q_{14}:=2C^2_2\,\kappa_1\left(\kappa^2_1-\kappa^2_2\right)^2\left(2\kappa^2_1-C^2_1x\right)\left[2\kappa_2x
-\sinh(2\kappa_2x)\right],
\end{equation*}
\begin{equation*}
Q_{15}:=C^2_1\,C^2_2\,\sinh(2\kappa_1x)\left[2\kappa_2\left(\kappa^2_1-\kappa^2_2\right)^2 x
-\left(\kappa^4_1+6\kappa^2_1\kappa^2_2+\kappa^4_2\right)\sinh(2\kappa_2x)\right],
\end{equation*}
\begin{equation*}
Q_{16}:=4\kappa^3_2\left(\kappa_1^2-\kappa_2^2\right)^2\left[4\kappa^3_1-2
\,C^2_1\,\kappa_1x+C^2_1\sinh(2\kappa_1x)\right].
\end{equation*}
Then, using \eqref{5.29}, \eqref{5.37}, \eqref{5.39}, and \eqref{5.40} in \eqref{4.12}, we obtain the quantity
$\tilde\alpha(x,y)$ as
\begin{equation}
\label{5.41}
\tilde\alpha(x,y)=\ds{\frac{-Q_{17}+(Q_{18}+Q_{19})\,(Q_{20}+Q_{21})}{\big[4\kappa^3_1-2\,C^2_1\kappa_1x+
C^2_1\,\sinh(2\kappa_1x)\big]\big[Q_{13}-Q_{22}-Q_{23}\big]}},
\end{equation}
where we have let
\begin{equation*}
Q_{17}:=4C^2_1\,\kappa_1\,\sinh(\kappa_1x)\,\sinh(\kappa_1y),
\quad
Q_{18}:=16C^2_1\,C^2_2\,\kappa_1\,\kappa_2^2\,\sinh^2(\kappa_1x)\cosh(\kappa_2x),
\end{equation*}
\begin{equation*}
Q_{19}:=\sinh(\kappa_2x)\left[8C^2_2\,\kappa_1\,\kappa_2\,(\kappa^2_1-\kappa^2_2)(2\kappa^2_1-C^2_1x)
-4C_1^2\,C^2_2\,\kappa_2(\kappa^2_1+\kappa^2_2)\sinh(2\kappa_1x)\right],
\end{equation*}
\begin{equation*}
Q_{20}:=4C^2_1\,\kappa_1\,\sinh(\kappa_1y)\left[\kappa_2\,\sinh(\kappa_1x)\cosh(\kappa_2x)
-\kappa_1\,\cosh(\kappa_1x)\sinh(\kappa_2x)\right],
\end{equation*}
\begin{equation*}
Q_{21}:=\left(\kappa^2_1-\kappa^2_2\right)\sinh(\kappa_2y)\left[4\kappa^3_1-2C^2_1\kappa_1x+C^2_1\sinh(2\kappa_1x)\right],
\end{equation*}
\begin{equation*}
 Q_{22}:=2\kappa_1(\kappa^2_1-\kappa^2_2)^2(2\kappa^2_1-C^2_1x)\left[4\kappa^3_2-2C^2_2\kappa_2x+C^2_2\sinh(2\kappa_2x)\right].
\end{equation*}
%\begin{equation*}
%Q_{22}:=\,16C^2_1\,C^2_2\,\kappa_1\,\kappa^3_2\,\cosh^2(\kappa_2x)\sinh^2(\kappa_1x),
%\quad
%Q_{23}:=16C^2_1\,C^2_2\,\kappa^3_1\,\kappa_2\,\cosh^2(\kappa_1x)\,\sinh^2(\kappa_2x),
%\end{equation*}
\begin{equation*}
Q_{23}:=C^2_1\sinh(2\kappa_1x)\left[2\kappa_2(\kappa^2_1-\kappa^2_2)^2(2\kappa^2_2-C^2_2x)+C^2_2(\kappa^4_1+
6\kappa^2_1\,\kappa^2_2+\kappa^4_2)\sinh(2\kappa_2x)\right],
\end{equation*}
%\begin{equation*}
% Q_{25}:=2\kappa_1(\kappa^2_1-\kappa^2_2)^2(2\kappa^2_1-C^2_1x)\left[4\kappa^3_2-2C^2_2\kappa_2x+C^2_2\sinh(2\kappa_2x)\right].
%\end{equation*}
Using $\tilde\alpha(x,y)$ given in \eqref{5.41} as input to \eqref{4.31},
we evaluate the resulting integral in \eqref{4.31} explicitly in a closed form, and hence we obtain the
perturbed wavefunction $\tilde\varphi(k,x)$ explicitly in terms of elementary functions. Since the resulting
explicit expression is very lengthy,
we do not quote it here but mention that a symbolic
software such as Mathematica displays
that lengthy expression explicitly.
We then evaluate $\tilde\alpha(x,x)-\alpha(x)$ by using $y=x$ in \eqref{5.29} and \eqref{5.41}.
The use of that resulting explicit expression in \eqref{5.27} yields
the potential perturbation $\tilde u(x)-u(x).$ Since $u(x)$ is already known explicitly, 
we also obtain the perturbed potential $\tilde u(x)$ explicitly in terms of elementary functions.
The resulting
explicit expression for $\tilde u(x)$ is also very lengthy,
and hence we do not display it here but mention again that it can be displayed
explicitly with the help of Mathematica. 

\end{example}

In the next example, we illustrate our generalized method on
 the interval
$(0,x)$ by considering
the Darboux transformation for 
the half-line Schr\"odinger equation \eqref{4.14} with a non-Dirichlet boundary condition.

\begin{example}
\label{example5.9}
\normalfont
In this example, we obtain the Darboux transformation at the potential and wavefunction levels
by adding one bound state to the half-line Schr\"odinger equation \eqref{4.14}
with the non-Dirichlet boundary condition \eqref{4.33} with $\cot\theta=-\kappa_1.$
As input to the unperturbed fundamental integral equation \eqref{4.1}, let us use
the operator kernel $\omega(x,y)$ given by
\begin{equation}
\label{5.42}
\omega(x,y)=-\ds\frac{\kappa_1}{2}\big[e^{-\kappa_1(x+y)}+e^{-\kappa_1(x-y)}\big],\qquad 0< y< x,
\end{equation}
where $\kappa_1$ is a positive parameter. Using \eqref{5.42} in \eqref{4.1}, we evaluate the solution $\alpha(x,y)$
to \eqref{4.1}
as
\begin{equation}
\label{5.43}
\alpha(x,y)=\kappa_1,\qquad 0< y<x.
\end{equation}
Let us use \eqref{5.43} in \eqref{4.4} in order to determine the corresponding resolvent kernel
$r(x;z,y).$
Since the quantity $\alpha(x,y)$ in \eqref{5.43}
is real valued and a scalar, as in Example~\ref{example5.7} we know that 
\eqref{4.4} is equivalent to \eqref{5.30}. Hence, using \eqref{5.43} in \eqref{5.30}, we obtain
the resolvent kernel
$r(x;z,y)$ as
\begin{equation}
\label{5.44}
r(x;z,y)=\begin{cases}\kappa_1+\kappa_1^2(x-z), \qquad 0< y< z<x,\\
\noalign{\medskip}
\kappa_1+\kappa_1^2(x-y),\qquad 0< z< y<x.\end{cases}
\end{equation}
One can directly verify that \eqref{4.3} holds when we use as input to \eqref{4.3} 
the
quantities $\omega(x,y),$ $\alpha(x,y),$ and $r(x;z,y)$
given in \eqref{5.42}, \eqref{5.43}, and \eqref{5.44}, respectively.
On the other hand, analogous to the illustration in Example~\ref{example5.2}, let us remark that the integral
equation \eqref{4.3}
may still be satisfied if we use as input to \eqref{4.3} the quantities
$\omega(x,y)$ of \eqref{5.42} and $\alpha(x,y)$ of
\eqref{5.43} and some quantity for $r(x;z,y)$
other than \eqref{5.44}. For example, let us consider the function
\begin{equation}
\label{5.45}
r(x;z,y)=\ds\frac{\kappa_1}{\sinh(\kappa_1x)}\big[e^{\kappa_1x}-\cosh(\kappa_1y)\big],
\end{equation}
which is not the resolvent kernel corresponding to \eqref{5.42} and \eqref{5.43}.
One can verify that \eqref{4.3} remains satisfied 
when the triplet \eqref{5.42}, \eqref{5.43}, \eqref{5.45} is used as input instead of the triplet
\eqref{5.42}, \eqref{5.43},  \eqref{5.44}.
In this example, the relevant unperturbed wavefunction $\varphi(k,x)$ related to
$\alpha(x,y)$ of \eqref{5.43} is obtained from the first equality of \eqref{4.36}, and we get
\begin{equation}
\label{5.46}
\varphi(k,x)=\cos(kx)+\kappa_1\,\frac{\sin(kx)}{k}.
\end{equation}
The unperturbed wavefunction $\varphi(k,x)$ in \eqref{5.46}
is the regular solution 
to \eqref{4.14} satisfying the initial conditions in \eqref{4.34} with
$\cot\theta=-\kappa_1,$ and it corresponds to
the zero potential $u(x)\equiv 0$ without any bound states.
The addition of a bound state to the unperturbed problem at $k=i\kappa_1$ with the norming constant $C_1$
is accomplished by choosing the perturbation kernel $\tilde\omega(x,y)-\omega(x,y)$ given
in \eqref{2.11} as
\begin{equation}
\label{5.47}
\tilde\omega(x,y)-\omega(x,y)=C_1^2\,\cosh(\kappa_1x)\,\cosh(\kappa_1y).
\end{equation}
Comparing \eqref{2.11} and \eqref{5.47}, we see that the quantities $f(x)$ and $g(y)$ in \eqref{2.11} can be chosen as
\begin{equation}
\label{5.48}
f(x)=C_1 \cosh(\kappa_1x), \quad  g(y)=C_1 \cosh(\kappa_1y).
\end{equation}
Using $\alpha(x,y)$ from \eqref{5.43} and $f(x)$ and $g(y)$ from \eqref{5.48}, we construct the 
intermediate quantities $n(x)$ and
$q(y)$ defined in \eqref{4.5} and \eqref{4.6}, respectively, as
\begin{equation}
\label{5.49}
n(x)=C_1\cosh(\kappa_1x)+C_1 \sinh(\kappa_1x),
\end{equation}
\begin{equation}
\label{5.50}
q(y)=C_1 \cosh(\kappa_1y)+C_1 \sinh(\kappa_1y).
\end{equation}
Then, we construct the intermediate quantity $\tilde g(x,y)$ 
by using \eqref{5.43} and \eqref{5.50} in \eqref{4.11}, and we obtain
\begin{equation}
\label{5.51}
\tilde g(x,y)=C_1\,e^{\kappa_1x},
\end{equation}
which does not depend on $y.$
Next, using \eqref{5.48} and \eqref{5.51} in \eqref{4.9}, we construct $\Gamma(x)$ as
\begin{equation}
\label{5.52}
\Gamma(x)=1+\ds\frac{C_1^2\,e^{\kappa_1x}\sinh(\kappa_1x)}{\kappa_1}.
\end{equation}
Finally, with the help of \eqref{5.43}, \eqref{5.49}, \eqref{5.51}, and \eqref{5.52}, from \eqref{4.12} we obtain
$\tilde\alpha(x,y)$ as
\begin{equation}
\label{5.53}
\tilde\alpha(x,y)=\ds\frac{\kappa_1\big[\kappa_1-C_1^2\,e^{\kappa_1x}\cosh(\kappa_1x)\big]}
{\kappa_1+C_1^2\,e^{\kappa_1x}\sinh(\kappa_1x)},\qquad 0<y<x,
\end{equation}
which is also independent of $y$ as a result of the $y$-independence in \eqref{5.51}.
The perturbed wavefunction $\tilde\varphi(k,x)$ is obtained by using
\eqref{5.53} in the second equality of \eqref{4.36}, and we have
\begin{equation}
\label{5.54}
\tilde\varphi(k,x)=\cos(kx)+\ds\frac{\sin(kx)}{k}\,\ds\frac{\kappa_1\,\big[\kappa_1-C_1^2\,e^{\kappa_1x}\cosh(\kappa_1x)\big]}{\kappa_1+C_1^2
\,e^{\kappa_1x}\sinh(\kappa_1x)},
\end{equation}
which satisfies the initial conditions
\begin{equation*}
\tilde\varphi(k,0)=1,\quad \tilde\varphi'(k,0)=\kappa_1-C_1^2.
\end{equation*}
Hence, for the corresponding perturbed Schr\"odinger operator, the value of $\cot\tilde\theta$ appearing in \eqref{4.35} is equal to $C_1^2-\kappa_1.$
In this case, the perturbation $\tilde u(x)-u(x)$ of the potential is obtained via \eqref{5.27}.
Using \eqref{5.43}, \eqref{5.53}, and the fact that $u(x)\equiv 0,$ from \eqref{5.27} we obtain the perturbed potential as
\begin{equation}
\label{5.55}
\tilde u(x)=\ds\frac{2C_1^2\,\kappa_1^2\,e^{2\kappa_1x}(C_1^2-2\kappa_1)}{\big[\kappa_1+C_1^2\,e^{\kappa_1x}\sinh(\kappa_1x)\big]^2}.
\end{equation}
From \eqref{5.54} and \eqref{5.55} we observe the following.
If the norming constant $C_1$ is chosen as $\sqrt{2 \kappa_1},$
then the perturbed potential is given by $\tilde u(x)\equiv 0$ and the perturbed wavefunction
$\tilde\varphi(k,x)$ becomes
\begin{equation*}
\tilde\varphi(k,x)=\cos(kx)-\kappa_1\,\frac{\sin(kx)}{k}.
\end{equation*}
This illustrates the fact that, for the Darboux transformation for the half-line Schr\"odinger equation, the relevant wavefunction to use
cannot be a wavefunction specified by any asymptotic condition at $x=+\infty.$ This is because such a wavefunction remains unchanged
under the Darboux transformation even though a bound state is added. In particular, the half-line Jost solution
to \eqref{4.15} remains unchanged under this particular Darboux transformation.
The relevant wavefunction to use in this case
must be a wavefunction specified by some initial conditions at $x=0.$
\end{example}

\end{document}